\documentclass[11pt]{amsart}
\pdfoutput=1
\usepackage[utf8]{inputenc}
\usepackage[T1]{fontenc}
\usepackage[a4paper,text={460pt,660pt},headsep=8mm, centering,marginparwidth=2cm]{geometry}
\usepackage[hidelinks,pdfusetitle]{hyperref}
\usepackage[activate={true,nocompatibility},final,tracking=true,kerning=true,factor=1100,stretch=10,shrink=10]{microtype}
\usepackage{braket}    
\usepackage{amsthm}    
\usepackage{amssymb}   
\usepackage{mathrsfs}  
\usepackage[shortlabels]{enumitem}
\usepackage{caption}
\usepackage[margin=5pt,justification=centering,labelformat=simple,labelfont=sc]{subcaption}
\usepackage[mode=buildnew]{standalone} 
\usepackage{tikz-cd}
\usepackage{pgfplots}
\pgfplotsset{compat=newest}

\usepackage{bbold}

\theoremstyle{plain}
\newtheorem*{theo*}{Theorem}
\newtheorem{theorem}{Theorem}[section]

\newtheorem{proposition}[theorem]{Proposition}
\newtheorem{lemma}[theorem]{Lemma}

\newtheorem{definition}[theorem]{Definition}
\newtheorem{hyp}{Assumption}

\theoremstyle{remark}
\newtheorem{remark}[theorem]{Remark}
\newtheorem*{rem*}{Remark}

\DeclareMathOperator*{\essinf}{ess\,inf}      
\DeclareMathOperator\spec{Spec}               

\newcommand{\cb}{\ensuremath{\mathscr{B}}}

\newcommand{\cf}{\ensuremath{\mathscr{F}}}
\newcommand{\cg}{\ensuremath{\mathscr{G}}}

\newcommand{\ck}{\ensuremath{\mathscr{K}}}
\newcommand{\ind}[1]{\mathbb{1}_{#1}}
\newcommand{\un}{\mathbb{1}}
\newcommand{\zero}{\mathbb{0}}
\newcommand{\ca}{\ensuremath{\mathscr{A}}}
\newcommand{\cl}{\ensuremath{\mathscr{L}^\infty}}
\newcommand{\gxx}{$[(\Omega, \cf, \mu), \mathrm{k}]$}

\usepackage[backend=biber,
	    url=false,
	    doi=false,
	    isbn=false,
	    giveninits=true,
	    date=year,
	    maxbibnames=99,
	    sortcites=true,
	    sorting=nty,
	    style=numeric]{biblatex}
\renewbibmacro{in:}{}
\usepackage{bm} 
\usepackage{dsfont}  

\date{\today}
\author{Jean-François Delmas}
\address{Jean-François Delmas,
  CERMICS, \'{E}cole nationale des ponts et chauss\'ees, 77455
  Marne-la-Vall\'ee, France} 
\email{jean-francois.delmas@enpc.fr}

\author{Dylan Dronnier}
\address{Dylan Dronnier,
  CERMICS, \'{E}cole nationale des ponts et chauss\'ees, 77455
  Marne-la-Vall\'ee, France} 
\email{dylan.dronnier@enpc.fr}

\author{Pierre-André Zitt}
\address{Pierre-André Zitt, LAMA, Université Gustave Eiffel, 77420
Champs-sur-Marne, France}
\email{pierre-andre.zitt@univ-eiffel.fr}

\newcommand{\abs}[1]{\left\lvert\,#1\,\right\rvert}
\newcommand{\norm}[1]{\left\lVert\,#1\,\right\rVert}

\newcommand{\R}{\ensuremath{\mathbb{R}}}

\newcommand{\C}{\ensuremath{\mathbb{C}}}
\newcommand{\N}{\ensuremath{\mathbb{N}}}

\newcommand{\kk}{\ensuremath{\mathrm{k}}}
\newcommand{\rd}{\ensuremath{\mathrm{d}}}
\newcommand{\cll}{\ensuremath{\mathcal{L}}}

\newcommand{\Tinf}{\ensuremath{\mathcal{T}}}

\newcommand{\costu}{\ensuremath{C_\mathrm{uni}}}

\newcommand{\oa}{\ensuremath{\Omega_\mathrm{a}}}
\newcommand{\oaa}{\ensuremath{\mathrm{a}}} 
\newcommand{\oi}{\ensuremath{\Omega_\mathrm{i}}}
\newcommand{\vg}{\ensuremath{v_\mathrm{g}}}
\newcommand{\vd}{\ensuremath{v_\mathrm{d}}}
\newcommand{\I}{\ensuremath{\mathfrak{I}}}

\newsavebox{\largestimage}

\usepgfplotslibrary{colorbrewer}
\pgfplotsset{colormap/Paired-6}

\newcommand{\loss}{\ensuremath{\mathrm{L}}}

\newcommand{\param}{\ensuremath{\mathrm{Param}}}

\newcommand{\gxxx}{$[(\Omega, \cf, \mu), k, \gamma]$}
\newcommand{\grR}{R_e[\kk]}
\newcommand{\grS}{\spec[\kk]}

\newcommand{\ka}{\ensuremath{\kk_\mathrm{a}}}
\newcommand{\cfi}{\ensuremath{\cf_\mathrm{inv}}}
\newcommand{\co}{\mathcal O}

\usepackage[draft]{fixme}
\FXRegisterAuthor{pa}{apa}{PAZ}
\FXRegisterAuthor{jf}{ajf}{JFD}
\fxusetheme{colorsig}

\newcommand{\Deltad}{\ensuremath{\mathbf{\Delta}}}
\newcommand{\Deltac}{\ensuremath{\mathit{\Delta}}}

\title[%
  Topological  properties of the effective reproduction number%
  ]{%
  Topological properties of the effective reproduction number in
  an heterogeneous  SIS model%
}

\begin{document}

\thanks{This work is partially supported by Labex Bézout reference
  ANR-10-LABX-58 and  SNF 200020-19699}

\subjclass[2010]{92D30, 58E17, 47B34, 34D20}

\keywords{SIS Model, infinite-dimensional ODE, vaccination strategy,
 effective reproduction number, kernel operator, irreducibility}

\begin{abstract}
 This present results lay the foundations  for the study of the optimal
  allocation of  vaccine in the  simple epidemiological SIS  model where
  one consider a very general  heterogeneous population.  In the present
  setting each  individual has a  type $x$ belonging to a  general 
  space,  and  a  vaccination  strategy   is  a  function  $\eta$  where
  $\eta(x)\in [0, 1]$ represents  the proportion of non-vaccinated among
  individuals  of  type  $x$.   We shall  consider  two  loss  functions
  associated  to a  vaccination  strategy $\eta$:  either the  effective
  reproduction number, a classical quantity  appearing in many models in
  epidemiology, and  which is  given here  by the  spectral radius  of a
  compact  operator  that depends  on~$\eta$; or  the
  overall  proportion  of  infected individuals  after
  vaccination in  the maximal 
  endemic state.

  By considering the weak-* topology on the set $\Deltad$ of vaccination
  strategies, so that it  is a compact set, we can  prove that those two
  loss  functions   are  continuous  using  the   notion  of  collective
  compactness for a family of  operators.  We also prove their stability
  with respect  to the parameters  of the  SIS model. Eventually, we  consider their
  monotonicity and  related properties in  particular when the  model is
  ``almost'' irreducible.
\end{abstract}

\maketitle

\section{Introduction}

\subsection{Motivation}

Increasing  the prevalence  of  immunity from  contagious  disease in  a
population limits the circulation of the infection among the individuals
who lack  immunity. This so-called  ``herd effect'' plays  a fundamental
role in epidemiology as it has had  a major impact in the eradication of
smallpox  and  rinderpest  or  the  near  eradication  of  poliomyelitis
\cite{HerdImmunityFine2011}. It is of course unrealistic to depict human
populations as homogeneous, and  many generalizations of the homogeneous
model have been studied; see \cite[Chapter 3]{keeling_modeling_2008} for
examples and further references.  Targeted vaccination strategies, based
on the heterogeneity  of the infection spreading in  the population, are
designed to  increase the  level of  immunity of  the population  with a
limited quantity of vaccine. These strategies rely on identifying groups
of individuals  that should be vaccinated  in priority in order  to slow
down  or eradicate  the  disease. It is assumed the vaccine is perfect
and provide an ever lasting immunity.

In  this  article,  we  consider  two  loss  functions  to  measure  the
effectiveness of targeted vaccination strategies with perfect vaccine in
the deterministic infinite-dimensional SIS model (with S=Susceptible and
I=Infectious)  introduced   in  \cite{delmas_infinite-dimensional_2020},
that  encompasses as  particular cases  the SIS  model on  graphs or  on
stochastic block models.

The  first one  is  the so-called  \emph{effective reproduction  number}
$R_e$  defined  as  the  number   of  secondary  cases  one  ``typical''
infectious  individual  generates on  average  over  the course  of  its
infectious  period,   in  an  otherwise  uninfected   (susceptible)  and
non-vaccinated population.   When there is no  vaccination, this reduces
to the  \emph{basic reproduction number} denoted  by~$R_0$.  This latter
number plays a  fundamental role in epidemiology as it  provides a scale
to  measure  how   difficult  an  infectious  disease   is  to  control,
see~\cite{Diekmann1990}.   Intuitively,  the   disease  should  die  out
if~$R_0<1$ (sub-critical  regime) and  invade the  population if~$R_0>1$
(super-critical  regime).  For  many  classical  mathematical models  of
epidemiology, such  as SIS or  S(E)IR (with R=Recovered  and E=Exposed),
this intuition can be made  rigorous: the quantity~$R_0$ may be computed
from the parameters of the model, and the threshold phenomenon occurs.

The  second  one  is  the  fraction  $\I$  of  infected  individuals  at
equilibrium,  and  set   $\I_0$  when  there  is   no  vaccination.  (In
particular, one get for the SIS  model that $\I_0=0$ in the sub-critical
regime $R_0\leq  1$.) For  a SIR  model, distributing  vaccine so  as to
minimize the  attack rate (that  is, the proportion of  individuals that
eventually catch (and recover from) the  disease) is at least as natural
as trying  to minimize  the reproduction number;  this problem  has been
studied       for        example       in~\cite{TheMostEfficiDuijze2016,
  DoseOptimalVaDuijze2018}.

\medskip

  The simplicity of the  SIS model allows us
to study the  regularity of the  loss  functions $r_e$ and $\I$ under minimal
assumptions for  general non-homogeneous populations,  using theoretical
properties on the  spectral radius of integral  operators and properties
of the maximal equilibrium of the SIS infinite dimensional ODE.
The mathematical  foundation developed here  allows us to  study Pareto optimal
vaccination in SIS model in~\cite{ddz-theory-optim}, when taking into
account the cost a the vaccination strategy,  and illustrate those
results     in    particular     cases     and    specific     examples,
see references therein. 
Furthermore,   we  expect   the
results obtained  for the  SIS model  to be generic,  in the  sense that
behaviours exhibited here should be  also observed in more realistic and
complex models in epidemiology  for non-homogeneous populations; in this
direction, see for example the discussion in \cite{ddz-Re}.

\subsection{Main results}

The   differential  equations   governing  the   epidemic  dynamics   in
metapopulation~SIS  models were  developed by  Lajmanovich and  Yorke in
their     pioneer     paper~\cite{lajmanovich1976deterministic}.      In
\cite{delmas_infinite-dimensional_2020},   we   introduced   a   natural
generalization of their equation, to a possibly infinite space~$\Omega$,
where~$x  \in   \Omega$  represents   a  feature  and   the  probability
measure~$\mu(\mathrm{d} x)$  represents the  fraction of  the population
with
feature~$x$. Following~\cite[Section~5]{delmas_infinite-dimensional_2020},
we    represent    a    vaccination    strategy    by    a    measurable
function~$\eta: \Omega  \rightarrow [0, 1]$,  where~$\eta(x)$ represents
the fraction of \textbf{non-vaccinated} individuals with feature~$x$. In
particular, the ``strategy'' that consists in vaccinating no one (resp.\
everybody) corresponds to  $\eta = \un$, the constant  function equal to
1, (resp.\ $\eta = \zero$, the  constant function equal to 0). We denote
by~$\Deltac$ the set of strategies.

\subsubsection{Regularity of the effective reproduction function~$R_e$}

We consider  the effective reproduction  function in a  general operator
framework which we call the  \emph{kernel model}. This model, which will
be   defined  in   detail   below   in  Section~\ref{sec:settings},   is
characterized by  a measured  space~$(\Omega, \cf,  \mu)$, with  $\mu$ a
non-zero   $\sigma$-finite  measure,   and  a   measurable  non-negative
kernel~$\kk:  \Omega \times  \Omega  \to \R_+$.  Considering the  kernel
model with a  general measure $\mu$ instead of a  probability measure is
in particular motivated by~\cite{ddz-Re, ddz-cordon}. Let $T_\kk$ be the
corresponding integral operator defined on some linear subspace of real-valued
measurable functions by:
\[
  T_\kk(h): \, x\mapsto \,  \int_\Omega \kk(x,y) h(y) \, \mu(\mathrm{d}y).
\]
In  the   setting  of~\cite{delmas_infinite-dimensional_2020}   (see  in
particular Equation~(11)  therein), $T_\kk$ is the  so-called \emph{next
  generation operator}, where the kernel~$\kk$  is defined in terms of a
transmission rate kernel $k(x,y)$  and a recovery rate function~$\gamma$
by   the  product~$\kk(x,y)=k(x,y)/\gamma(y)$;   the  \emph{reproduction
  number}~$R_0$    is   then    the   spectral    radius   $\rho(T_\kk)$
of~$T_\kk$.

\medskip

The effective
reproduction number associated to the vaccination strategy~$\eta\in
\Deltac$ is given by:
\begin{equation}\label{eq:def-Re-intro}
  R_e(\eta) = \rho(T_{\kk \eta}),
\end{equation}
where~$\rho$ stands for the spectral radius and~$\kk\eta$ stands for the kernel~$(\kk
\eta)(x,y)=\kk(x,y) \eta(y)$. 
In particular, we have~$R_e(\un) =
R_0$ (resp.\ $R_e(\zero) = 0$).

\medskip

Motivated by vaccine allocation optimization, we shall consider a
topology on $\Deltac$ such that it is compact and the function~$R_e$
is continuous.  It is natural to try and prove this continuity by
writing~$R_e$ as the composition of the spectral radius~$\rho$ and the
map~$\eta \mapsto T_{\kk\eta}$.  The spectral radius is indeed
continuous at compact operators (and $T_{\kk \eta}$ is in fact compact
under a technical integrability assumption on the kernel~$\kk$
formalized on page~\pageref{hyp:k} as Assumption~\ref{hyp:k}), if we
endow the set of bounded operators with the operator norm topology;
see \cite{newburgh1951, burlando}.  However, this only works if we
equip~$\Deltac$ with the uniform topology, for which it is not
compact.

We instead consider $\Deltad$, the set of functions in~$\Deltac$ where
functions which are $\mu$-a.e.\ equal are identified, endowed with the
weak-* topology for which compactness holds; see
Lemma~\ref{lem:D-compact}.  This forces us to equip the space of
bounded operators with the strong topology, for which the spectral
radius is in general not continuous
\cite[p.~431]{kato2013perturbation}.  However, the family of
operators~$(T_{\kk\eta}, \, \eta \in \Deltac)$ is \emph{collectively
  compact} which enables us to recover continuity, using a series of
results obtained by Anselone~\cite{anselone}.  After noticing that the
function $R_e$ coincide on functions which are $\mu$-a.e.\ equal, o
that $R_e$ is indeed well defined on $\Deltad$, this leads to the
following statement, proved in Theorem~\ref{th:continuity-R} below.
We recall that Assumption~\ref{hyp:k}, formulated on
page~\pageref{hyp:k}, provides an integrability condition on the
kernel~$\kk$.

\begin{theorem}[Continuity of the spectral radius]
  Under    Assumption~\ref{hyp:k}     on    the     kernel~$\kk$,    the
  function~$R_e  \,  \colon \,  \Deltad  \to  \R_+$ is  continuous  with
  respect to the weak-* topology on~$\Deltad$.
\end{theorem}

In fact,  we also prove the  continuity of the spectrum  with respect to
the Hausdorff distance  on the set of compact subsets  of~$\C$. We shall
write~$R_e[\kk]$ to stress  the dependence of the  function~$R_e$ in the
kernel~$\kk$. In Proposition \ref{prop:Re-stab},  we prove the stability
of~$R_e$,  by giving  natural  sufficient conditions  on  a sequence  of
kernels~$(\kk_n,   n\in    \N)$   converging   to~$\kk$    which   imply
that~$R_e[\kk_n]$  converges uniformly  towards~$R_e[\kk]$. This  result
has both  a theoretical  and a  practical interest:  the next-generation
operator is  unknown in  practice, and  has to  be estimated  from data.
Thanks to  this result, the value  of $R_e$ computed from  the estimated
operator is a converging approximation of the true value.

\subsubsection{Regularity of the total proportion of infected population function~$\I$}

We consider the \emph{SIS model} from
\cite{delmas_infinite-dimensional_2020}. This model is characterized
by a probability space $(\Omega, \cf,\mu)$, the transmission
kernel~$k \, \colon \, \Omega \times \Omega \to \R_+$ and the recovery
rate~$\gamma \, \colon \, \Omega \to \R_+^*$. We suppose in the
following that the technical Assumption~\ref{hyp:k-g}, formulated on
page~\pageref{hyp:k-g}, holds, so that the SIS dynamical evolution is
well defined.

This evolution is encoded as~$u=(u_t, t\in \R_+)$,
where~$u_t\in \Deltac$ for all~$t$ and $u_t(x)$ represents the
probability of an individual with feature~$x\in \Omega$ to be infected
at time~$t\geq 0$, and follows the equation:
\begin{equation}
  \label{eq:SIS-intro}
  \partial_t u_t = F(u_t)\quad\text{for } t\in \R_+,
  \quad\text{where}\quad F(g) = (\un - g) \Tinf_k (g) - \gamma g\quad\text{for } g\in \Deltac,
\end{equation}
with  an initial  condition~$u_0  \in \Deltac$  and  with $\Tinf_k$  the
integral operator corresponding  to the kernel $k$ acting on  the set of
bounded  measurable  functions,  see \eqref{eq:def-Tk}.   It  is  proved
in~\cite{delmas_infinite-dimensional_2020}  that   such  a  solution~$u$
exists    and   is    unique    under   Assumption~\ref{hyp:k-g}.     An
\emph{equilibrium} of~\eqref{eq:SIS-intro} is a function~$g \in \Deltac$
such       that       $F(g)        =       0$.        According       to
\cite{delmas_infinite-dimensional_2020},   there    exists   a   maximal
equilibrium~$\mathfrak{g}$, \textit{i.e.}, an  equilibrium such that all
other     equilibria      $h     \in     \Deltac$      are     dominated
by~$\mathfrak{g}$:~$h     \leq    \mathfrak{g}$.      Furthermore,    we
have~$R_0\leq  1$  (sub-critical  and  critical  regimes)  if  and  only
if~$\mathfrak{g}=0$. In  the non-trivial connected case  (for example if
$k>0$),  then~$\zero$ and~$\mathfrak{g}$  are the  only equilibria,  and
$\mathfrak{g}$ is the long-time  distribution of infected individuals in
the population:~$\lim_{t\rightarrow+\infty } u_t = \mathfrak{g}$ as soon
as      the      initial       condition      is      non-zero;      see
\cite[Theorem~4.14]{delmas_infinite-dimensional_2020}.

 \medskip

 According to \cite[Section~5.3]{delmas_infinite-dimensional_2020},
 the SIS equation with vaccination strategy~$\eta$ is given
 by~\eqref{eq:SIS-intro}, where~$F$ is replaced by $F_\eta$ defined
 by:
\[
  F_\eta(g) = (\un-g)T_{k\eta}(g) - \gamma g.
\]
and~$u_t$ now describes the proportion of infected \emph{among the
  non-vaccinated population}. We denote by $\mathfrak{g}_\eta$ the
corresponding maximal equilibrium (thus considering~$\eta=\un$ gives
$\mathfrak{g}=\mathfrak{g}_{\un}$), so
that~$F_\eta(\mathfrak{g}_\eta)=0$. Since the probability for an
individual~$x$ to be infected in the stationary regime
is~$\mathfrak{g}_\eta(x) \, \eta(x)$, the \emph{fraction of infected
  individuals at equilibrium},~$\I(\eta)$, is thus given by:
\begin{equation}%
  \label{eq:def-I-intro}
  \I (\eta)=\int_\Omega \mathfrak{g}_\eta
  \, \eta\, \mathrm{d}\mu=\int_\Omega \mathfrak{g}_\eta(x)
  \, \eta(x)\, \mu(\mathrm{d}x).
\end{equation}

In the SIS model the quantity~$\I$  appears as a natural analogue of the
attack  rate for  SIR models,  and is  therefore a  natural optimization
objective.

\medskip

We obtain results on the functional~$\I$ that are very similar to the
ones on~$R_e$. Recall that Assumption~\ref{hyp:k-g} on
page~\pageref{hyp:k-g} ensures that the infinite-dimensional SIS
model, given by equation \eqref{eq:SIS-intro}, is well defined. The
next theorem corresponds to Theorem~\ref{th:continuity-I}.

\begin{theorem}[Continuity of the equilibrium infection size]
  Under Assumption~\ref{hyp:k-g}, the
  function~$\I \, \colon \, \Deltad \to \R_+$ is continuous with
  respect to the weak-* topology on~$\Deltad$.
\end{theorem}

We shall write~$\I[k,\gamma]$ to stress the dependence of the
function~$\I$ in the kernel~$k$ and the function~$\gamma$. In
Proposition~\ref{prop:I-stab}, we prove the stability of~$\I$, by
giving natural sufficient conditions on a sequence of kernels and
functions $((k_n, \gamma_n), n\in \N)$ converging to~$(k, \gamma)$
which imply that~$\I[k_n, \gamma_n]$ converges uniformly
towards~$\I[k,\gamma]$.

\subsubsection{Other regularity results}
We also prove that the loss functions~$\loss=R_e$ and~$\loss=\I$ are
both non-decreasing ($\eta\leq \eta'$
implies~$\loss(\eta)\leq \loss(\eta')$), and sub-homogeneous
($ \loss(\lambda \eta)\leq \lambda \loss(\eta)$ for all
$\lambda\in [0,1]$); see Propositions~\ref{prop:R_e} and~\ref{prop:I}.

Motivated by the  bi-objective minimization problem of the  cost and the
loss  $\loss$  of vaccination  strategies  and  the description  of  the
corresponding set of Pareto  optimal vaccination strategies developed in
the  companion paper~\cite{ddz-theory-optim},  we shall  investigated if
local  extrema of  the loss  function are  in fact  global extrema,  see
Assumptions~\ref{hyp:loss-topo}     on     pages~\pageref{hyp:loss-topo}
and~\pageref{hyp:loss=I}.  It  turns out  that local minimum  are indeed
global  minimum for  the  loss  functions $R_e$  and  $\I$. However  the
picture is more  involved for the local maximum,  and slightly different
between $R_e$ and  $\I$. We concentrate in this paper  on the case where
the  model is  irreducible and  its extension,  the so  called monatomic
case,  where   intuitively,  there  is  only   one  maximal  irreducible
component.  Those results are  given in Lemmas~\ref{lem:L**} (for $R_e$)
and~\ref{lem:L**=I}  (for $\I$).   We also  characterize all  the global
maxima.  Let  us mention that the  reducible case for the  loss $R_e$ is
further studied in~\cite[Section~5]{ddz-cordon}.

\subsection{Structure of the paper}

After recalling a few topological facts in Section~\ref{sec:settings},
we present the vaccination model, the loss functions~$R_e$ and~$\I$,
and the various assumptions on the parameters in
Section~\ref{sec:settings}. We study the regularity properties
of~$R_e$ and~$\I$ in~Section~\ref{sec:properties-loss}.
Section~\ref{sec:other-reg} is devoted to study of their local
extremum.  The proofs of a few technical results on $\I$ are gathered
in Section~\ref{sec:technical_proofs}.

\section{General setting and notation}\label{sec:settings}

\subsection{Spaces, operators, spectra}\label{sec:spaces}

All metric spaces~$(S,d)$ are endowed with their Borel~$\sigma$-field
denoted by $\cb(S)$.  The set $\ck$ of compact subsets of~$\C$ endowed
with the Hausdorff distance $d_\mathrm{H}$ is a metric space, and the
function~$\mathrm{rad}$ from~$\ck$ to $\R_+$ defined
by~$\mathrm{rad}(K)=\max\{|\lambda|\, ,\, \lambda\in K\}$ is Lipschitz
continuous from~$(\ck,d_\mathrm{H})$ to~$\R$ endowed with its usual
Euclidean distance.

Let~$(\Omega, \cf)$ be a measurable space endowed with a
$\sigma$-finite non-negative measure $\mu\neq 0$.  We denote
by~$\mathscr{L}^\infty$, the Banach spaces of bounded real-valued
measurable functions defined on~$\Omega$ equipped with
the~$\sup$-norm, $\mathscr{L}^\infty_+$ the subset
of~$\mathscr{L}^\infty$ of non-negative function, and
$\Deltac=\{f\in \mathscr{L}^\infty\,\colon\, f(\Omega) \subset [0,
1]\}$ the subset of non-negative functions bounded by~$1$. For~$f$
and~$g$ real-valued functions defined on~$\Omega$, we may
write~$\langle f, g \rangle$ or $\int_\Omega f g \, \mathrm{d} \mu$
for $\int_\Omega f(x) g(x) \,\mu( \mathrm{d} x)$ whenever the latter
integral is meaningful.  For~$p \in [1, +\infty]$, we denote by
$L^p=L^p( \mu)=L^p(\Omega, \mu)$ the space of real-valued measurable
functions~$g$ defined~$\Omega$ such that
$\norm{g}_p=\left(\int |g|^p \, \mathrm{d} \mu\right)^{1/p}$ (with the
convention that~$\norm{g}_\infty$ is the~$\mu$-essential supremum of
$|g|$) is finite, where functions which agree~$\mu$-almost surely are
identified.  We denote by $\zero$ and $\un$ the elements of $\cl$
which are respectively the (class of equivalence of the) constant
functions equal to $0$ and to $1$, and with a slight abuse of
notation, we also see them as elements of $L^\infty $.  For
$f,g\in L^p$, the inequality $f\leq g$ (in $L ^p$) means that
$\mu(f>g)=0$.  We consider the Banach lattice
$(L^p, \norm{\cdot}_p, \leq )$ and its
cone~$L^p_+=\{f\in L^p \, \colon\, f\geq \zero\}$ of non-negative
functions from $L^p$.  We shall consider the set
$\Deltad=\{ f\in L^\infty \, \colon\, \zero \leq f \leq \un\}$
corresponding to the set $\Deltac$ where functions which agree
$\mu$-a.e.\ are identified.  For $g$ a measurable function, with a
slight abuse of notation, we denote by~$M_g$ the multiplication linear
map (possibly unbounded) on $L^p$ or on $\cl $ defined by~$M_g(h)=gh$.

\medskip

We now recall some general facts on Banach spaces and Banach lattices. 
  Let~$(E, \norm{\cdot})$ be a real  or complex Banach space.  We denote
  by~$\norm{\cdot}_E$ the operator norm  on~$\cll(E)$ the Banach algebra
  of operators,  that is, bounded  linear maps.  Let~$T\in  \cll(E)$. The spectral radius of~$T$ is
given by:
\begin{equation}\label{eq:def-rho-real}
  \rho(T)=
  \lim_{n\rightarrow \infty } \norm{T^n}_E^{1/n}.
\end{equation}
A sequence~$(T_n,n \in \N)$ of
elements      of~$\cll(E)$      converges     strongly      to~$T $
if~$\lim_{n\rightarrow \infty } \norm{T_nx -Tx}=0$ for all~$x\in E$. The
operator        $T        $        is       compact        if        the
subset~$\{ T  x \,  \colon \,  \norm{x}\leq 1 \}$  of $E$  is relatively
compact;      and      following~\cite{anselone},     a      set      of
operators~$\ca\subset  \cll(E)$ is  \emph{collectively  compact} if  the
subset~$\{ S x \,  \colon \, S \in \ca, \, \norm{x}\leq 1  \}$ of $E$ is
relatively compact.

If~$(E, \norm{\cdot})$ is a complex Banach space, the
spectrum~$\spec(T)$ of~$T\in \cll(E)$ is the set of~$\lambda\in \C$
such that~$ T - \lambda \mathrm{Id}$ does not have a bounded inverse
linear map, where~$\mathrm{Id}$ is the identity operator
on~$E$. Recall that~$\spec(T)$ is a compact subset of~$\C$, and that
the spectral radius of~$T$ is also given by:
\begin{equation}\label{eq:def-rho}
  \rho(T)=\mathrm{rad}(\spec(T)).
\end{equation}
The  element $\lambda\in  \spec(T)$  is an  eigenvalue  if there  exists
$x\in E$ such that $Tx=\lambda x$ and $x\neq 0$.

The next result is in~\cite{anselone} (see \cite[Lemma~2.1]{ddz-Re} for
details). 
 
\begin{lemma}[Anselone]
  \label{lem:prop-spec-mult}
 Let~$(T_n, n\in \N)$ be a collectively compact sequence of $\cll(E)$ which
 converges strongly to~$T\in \cll(E)$. Then, we have $\lim_{n\rightarrow \infty }
 \spec(T_n)=\spec(T)$ in~$(\ck, d_\mathrm{H})$, and $\lim_{n\rightarrow }
 \rho(T_n)=\rho(T)$.
\end{lemma}

Let  $(E, \norm{\cdot},  \leq)$  be  a  real Banach  lattice, that is
$(E,    \norm{\cdot})$ is a real Banach space with an order relation
$\leq $ satisfying some conditions, see 
\cite[Section~9.1]{aliprantis}. 
We  denote by
$E_+  =  \{x  \in E  \,\colon\,  x  \geq  0  \}$ the  positive  cone  of
$E$.  Recall it is a closed
set.  A  linear map  $T$ on  $E$ is  \emph{positive} if
$T(E_+)  \subset  E_+$.   According  to~\cite[Theorem~4.3]{aliprantis}
positive  linear maps  on  Banach  lattices are  bounded  (and thus  are
operators).
If  $S$ and $T$ are  two operators on $E$,
we write  $T \leq S$ if the operator $S - T$ is positive.
The next result can be found in~\cite[Theorem 4.2]{marek70}. 
\begin{lemma}
  \label{lem:prop-spec-mult-2}
  Let~$(E, \norm{\cdot},  \leq)$  be  a  real Banach  lattice.
  Let $S,T\in\cll(E)$ be positive operators. If $T\leq S$, then we have:
\begin{equation}\label{eq:spec_rad_croissant}
\rho(T) \leq \rho(S).
\end{equation}
\end{lemma}

Any real Banach lattice $E$ and any operator $T$ on $E$ admits a natural
complex  extension.  The  spectrum  of  $T$ will  be  identified as  the
spectrum of its complex extension and denoted by $\spec(T)$, furthermore
by \cite[Lemma  6.22]{abramovich02}, the spectral radius  of the complex
extension of $T$ is also equal to the spectral radius of $T$.  Moreover,
by \cite[Corollary 3.23]{abramovich02},  if $T$ is positive  (seen as an
operator  on the  real Banach  lattice $E$),  then $T$  and its  complex
extension have the same norm.

\subsection{On the weak-* topology on $\Deltad$}\label{sec:weak}

Let $\co$ denote the \emph{weak-* topology} on~$L^\infty $, that is,
the weakest topology on $L^\infty $ for which all the linear forms
$f\mapsto \int_\Omega fg \, \rd \mu$, $g\in L^1$, defined on
$L^\infty $ are continuous.  We recall that $(L^\infty , \co)$ is an
Hausdorff topological vector space,
see~\cite[Proposition~3.11]{brezis2010functional}, that the
topological dual of $(L^\infty , \co)$ is $L^1$,
see~\cite[Proposition~3.14]{brezis2010functional}, and that a
sequence~$(f_n, \, n \in \N)$ of elements of~$L^\infty $ converges
weakly-* to~$f\in L^\infty $ if and only if,
see~\cite[Proposition~3.13]{brezis2010functional}.

\begin{equation}\label{eq:weak-cv}
  \lim\limits_{n \to \infty} \int_\Omega g f_n \, \mathrm{d}\mu=
  \int_\Omega g f\, \mathrm{d}\mu \quad\text{for all~$g \in L^1 $.}
\end{equation}

A set  $A\subset L^\infty $  is weak-*  sequentially compact if  for all
sequences of elements of~$A$, there exists a sub-sequence which weakly-*
converges to a limit belonging to  $A$.  A topological set $(E, \co)$ is a
sequential space if  for any  $A\subset E$  which is  not closed,  there exist
$x\in \bar A \setminus  A$, where $\bar A$ is the closure  of $A$, and a
sequence in $A$ which converges to $x$.

\begin{lemma}[Topological properties of~$\Deltad\subset L^\infty $]\label{lem:D-compact}
Let   $(\Omega, \cf, \mu)$  be  a measured space with $\mu$ a
 $\sigma$-finite measure, and consider the weak-* topology on $L^\infty
 $. The following properties hold. 
 \begin{enumerate}[(i)]
 \item\label{item:D-compact} The set~$\Deltad$ is  weak-*  compact
 and weak-* sequentially compact.
   \item The set~$\Deltad$ endowed with the weak-* topology is a
     sequential space. 
\item\label{item:f-cont}  A function  from $\Deltad$  (endowed with  the
  weak-* topology) to  a topological space is continuous if  and only if
  it is sequentially continuous.

 \end{enumerate}
\end{lemma}

\begin{proof}
  The Banach-Alaoglu theorem \cite[Theorem~3.21]{fhhspz} implies that
  the closed unit ball, say $B_{L^\infty} $, of $L^\infty $ is weak-*
  compact.  According to \cite[Example~(v), Chapter~11]{fhhspz} as
  $\mu$ is $\sigma$-finite, the Banach space $L^1$ is weakly compactly
  generated (that is, there exists a weakly compact set $K$ whose
  linear span is dense in $L^1$).  Thus, thanks to the
  Amir-Lindenstraus theorem, see Theorem~11.16 or more directly
  Exercise~11.21 in \cite{fhhspz}, the unit ball $B_{L^\infty} $ is
  weak-* sequentially compact and in fact weak-* angelic (that is, for
  all $A\subset B_{L^\infty} $ and all $x$ in the weak-* closure of
  $A$, there exists a sequence of elements in $A$ which weak-*
  converges to $x$, see \cite[Definition~4.48]{fhhspz}).  In
  particular, since $\Deltad$ is the closed ball centered at
  $2^{-1} \un$ with radius $1/2$ of~$L^\infty$, we get it is weak-*
  compact, weak-* sequentially compact and weak-* angelic.

 \medskip

 Since $\Deltad$ is weak-* angelic, we deduce that it is a sequential
 space.  Since continuity and sequential continuity coincide for
 functions defined on a sequential space, we get~\ref{item:f-cont}.
\end{proof}

\begin{remark}[On the topology on $\Deltad$]
   \label{rem:topo}
   Assume that the  measure $\mu$ is finite.  Let $p\in  (1, +\infty )$.
   Using that reflexive Banach spaces are weakly compactly
   generated according to~\cite[Example~(i), Chapter~11]{fhhspz}, we
   get, arguing as  in the  proof of  Lemma~\ref{lem:D-compact}, 
   that the  set $\Deltad$ with  the trace  of the weak-*  topology (and
   thus of  the weak topology  as the space is  reflexive) on $L^p  $ is
   also    a  sequential   space.
   Furthermore,  with  $1/p+1/q=1$,  a
   sequence~$(f_n, \, n \in \N)$  of elements of~$L^p $ converges weakly
   to~$f\in L^p $ if and only if:
\begin{equation}
  \label{eq:weak-cv2}
  \lim\limits_{n \to \infty} \int_\Omega  gf_n \, \mathrm{d}\mu= \int_\Omega 
  gf\, \mathrm{d}\mu
  \quad\text{for
    all~$g \in L^q $.}
\end{equation}
Since the topology on a sequential 
space     is    characterized     by    the     converging    sequences,
see~\cite[Exercises~1.7.20]{engelking},   and   since~\eqref{eq:weak-cv}
and~\eqref{eq:weak-cv2}           are           equivalent           for
sequences~$(g_n, \, n \in \N)$  of elements of~$\Deltad$, we deduce that
the trace on~$\Deltad$ of the weak-*  topology on $L^\infty $ and of the
weak  topology on  $L^p$ coincide.  (Let us  stress that  there exists  a
topology  different  from  the  weak-*   topology  which  has  the  same
converging sequences, see the last proposition in~\cite{rubel}.)
\end{remark}

We shall consider loss functions $\loss$ defined on $\Deltac\subset \cl$, and see
them as function on $\Deltad\subset L^\infty $ when they are compatible
with the equivalence relation given by the $\mu$-a.e.\ equality. In this
case, with a slight abuse of notation, we also denote the corresponding
function on $\Deltad$ by $\loss$. 

\begin{definition}
   \label{defi:well-def}
A    loss
function $\loss$ defined   on $\Deltac$ is:
\begin{enumerate}[(i)]
   \item \textbf{Well defined}   (on $\Deltad$ endowed with the
weak-*   topology) if   for   all
$\eta_1,  \eta_2\in  \Deltac$:
\begin{equation}
  \label{eq:loss-defined}
  \eta_1=\eta_2 \quad \mu\text{-a.e.}\quad
   \implies\quad  \loss   (\eta_1)=\loss   (\eta_2);
\end{equation}

\item \textbf{Non-decreasing}   on $\Deltad$  if   for   all
$\eta_1,  \eta_2\in  \Deltac$:
\begin{equation}
  \label{eq:loss-increase}
  \eta_1\leq \eta_2 \quad \mu\text{-a.e.}\quad
   \implies\quad  \loss   (\eta_1)\leq \loss   (\eta_2);
\end{equation}
\item \textbf{Sub-homogeneous}   on $\Deltad$  if   for   all
$\eta\in  \Deltac$ and $\lambda\in [0, 1]$:
\begin{equation}
  \label{eq:sub-hom}
\loss(\lambda \eta) \leq  \lambda \, \loss(\eta). 
\end{equation}
\end{enumerate}
\end{definition}

 \section{The kernel and SIS models}
 
\subsection{Kernel model ($\mu(\Omega)\in (0, +\infty ]$)}
\emph{In the kernel model, we assume that the measure $\mu$ is $\sigma$-finite
and non-zero.}   We define a \emph{kernel}  (resp.\ \emph{signed kernel})
on~$\Omega$ as a $\R_+$-valued  (resp.\ $\R$-valued) measurable function
defined   on~$(\Omega^2,  \mathscr{F}^{\otimes   2})$.   For~$f,g$   two
non-negative measurable functions defined on~$\Omega$ and~$\kk$ a kernel
on~$\Omega$, we denote by $f\kk g$ the kernel on $\Omega$ defined by:
\begin{equation}
  \label{eq:def-fkg}
  f\kk g:(x,y)\mapsto f(x)\, \kk(x,y) g(y).
\end{equation}
When~$\gamma$ is a positive measurable function defined on~$\Omega$, we write~$\kk/\gamma$
for~$\kk\gamma^{-1}$, which  differs in general from~$\gamma^{-1} \kk$.

For~$p \in (1, +\infty )$, we define the double norm of a signed kernel~$\kk$ by:
\begin{equation}\label{eq:Lp-integ-cond}
  \norm{\kk}_{p,q}=\left(\int_\Omega\left( \int_\Omega \abs{\kk(x,y)}^q\,
  \mu(\mathrm{d}y)\right)^{p/q} \mu(\mathrm{d}x) \right)^{1/p}
  \quad\text{with~$q$ given by}\quad \frac{1}{p}+\frac{1}{q}=1.
\end{equation}

\begin{hyp}[On the kernel model \gxx]\label{hyp:k}
  The kernel~$\kk$, defined on a measured space~$(\Omega,\cf,\mu)$, with
  $\sigma$-finite   non-zero   measure   $\mu$,   has   a   \emph{finite
    double-norm},      that      is,~$\norm{\kk}_{p,q}<+\infty$      for
  some~$p\in (1, +\infty )$.
\end{hyp}

To a kernel $\kk$ such that $\norm{\kk}_{p,q}<+\infty$, we associate the integral
operator~$T_\kk$ on~$L^p$ defined by:
\begin{equation}\label{eq:def-Tkk}
  T_\kk (g) (x) = \int_\Omega \kk(x,y) g(y)\,\mu(\mathrm{d}y)
  \quad \text{for } g\in L^p \text{ and } x\in \Omega.
\end{equation}
This operator is positive (in the sense that $T_\kk(L^p_+) \subset
L^p_+$), and compact (see~\cite[p. 293]{grobler}). It is well known and easy to
check that:
\begin{equation}\label{eq:double-norm-norm}
  \norm{ T_\kk }_{L^p}\leq \norm{\kk}_{p,q}.
\end{equation}
For~$\eta\in \Deltac$, the kernel~$\kk \eta$ has also a finite double norm on~$L^p$ and the
operator~$M_\eta$ is bounded, so that the operator $T_{\kk \eta} = T_\kk M_\eta$ is
compact. We can define the \emph{effective spectrum} function~$\spec[\kk]$ from~$\Deltac$
to~$\ck$ by:
\begin{equation}\label{eq:def-sigma_e}
  \spec[\kk](\eta)=\spec(T_{\kk\eta}),
\end{equation}
the \emph{effective reproduction number} function $R_e[\kk]=\mathrm{rad}\circ \spec[\kk]$
from~$\Deltac$ to~$\R_+$ by:
\begin{equation}
  \label{eq:def-R_e}
  R_e[\kk](\eta)=\rho(T_{\kk\eta}),
\end{equation}
and the corresponding \emph{reproduction number}:
\begin{equation}\label{eq:def-R0}
  R_0[\kk]=R_e[\kk](\un)=\rho(T_\kk).
\end{equation}
When  there is  no ambiguity,  we simply  write $R_e$  for the  function
$R_e[\kk]$,  and   $R_0$  for  the  number   $R_0[\kk]$.   Motivated  by
Section~\ref{sec:vacc}   below,   we    say   a   vaccination   strategy
$\eta\in \Deltac$ is \emph{critical} if $R_e(\eta)=1$.

\subsection{SIS model ($\mu(\Omega)=1$): dynamics and equilibria}
\label{sec:dyn}
\emph{In  the  SIS  model,  we  assume that  $\mu$  is  a  probability
  measure},         thus         following         the         framework
of~\cite{delmas_infinite-dimensional_2020}. For $q\in (1, +\infty )$, we
also consider the following norm for the kernel $\kk$:
\[
  \norm{\kk}_{\infty,q} = \sup\limits_{x \in \Omega} \left(\int_\Omega \kk(x,y)^q\,
  \mu(\mathrm{d}y) \right)^{1/q}.
\]
Since $\mu$ is finite, if the  norm $\norm{\kk}_{\infty , q}$ is finite,
then for $p$ such that  $1/p+1/q=1$, the norm~$\norm{\kk}_{p,q}$ is also
finite.   When $\norm{\kk}_{\infty  ,  q}<+\infty  $, the  corresponding
positive      bounded      linear     integral      operator~$\Tinf_\kk$
on~$\mathscr{L}^\infty$ is similarly defined by:
\begin{equation}
  \label{eq:def-Tk}
  \Tinf_{\kk} (g) (x) = \int_\Omega \kk(x,y) g(y)\,
  \mu(\mathrm{d}y)
  \quad \text{for } g\in \mathscr{L}^\infty \text{ and } x\in \Omega.
\end{equation}
Notice that the integral operators $\Tinf_\kk$ and $T_\kk$ corresponds respectively to the
operators $T_\kk$ and $\hat T_\kk$ in \cite{delmas_infinite-dimensional_2020}. According
to \cite[Lemma~3.7]{delmas_infinite-dimensional_2020}, the operator~$\Tinf_\kk^2$ on~$\cl$
is compact and~$\Tinf_\kk$ has the same spectral radius as~$ T_\kk$:
\begin{equation}\label{eq:rhoT=rhoT}
  \rho(\Tinf_\kk)=\rho(T_\kk).
\end{equation}

In accordance with  \cite{delmas_infinite-dimensional_2020}, we consider
the following assumption. Recall that $k/\gamma=k \gamma^{-1}$.

\begin{hyp}[On the SIS model \gxxx]\label{hyp:k-g}
 The recovery rate
  function~$\gamma$,   defined on a probability space
  $(\Omega,\cf,\mu)$,
  is bounded and positive. 
    The transmission rate kernel~$k$ on $\Omega$
  is such that   $\norm{k/\gamma}_{\infty , q}<+\infty $ for some
  $q\in (1, +\infty )$.
  \end{hyp}

  If $k$ and $\gamma$ satisfy Assumption~\ref{hyp:k-g}, then
  $\kk = k/\gamma$ clearly satisfies Assumption~\ref{hyp:k} (as $\mu$ is
  finite).  Under
  Assumption~\ref{hyp:k-g}, we also consider the bounded
  operators~$\Tinf_{k / \gamma}$ on $\cl$, as well as $T_{k / \gamma}$
  on $L^p$, which are the so called \emph{next-generation
    operator}.
  The SIS dynamics considered
  in~\cite{delmas_infinite-dimensional_2020} under Assumption
  \ref{hyp:k-g} follows the vector field~$F$ defined
  on~$\mathscr{L}^\infty$ by:
\begin{equation}\label{eq:vec-field}
  F(g) = (\un - g) \Tinf_k (g) - \gamma g.
\end{equation}
More precisely, we consider~$u=(u_t, t\in \R)$, where~$u_t\in \Deltac$ for all~$t\in\R_+$
such that:
\begin{equation}\label{eq:SIS2}
  \partial_t u_t = F(u_t)\quad\text{for } t\in \R_+,
\end{equation}
with initial condition~$u_0\in \Deltac$. The value~$u_t(x)$ models the probability that an
individual of feature~$x$ is infected at time~$t$; it is proved
in~\cite{delmas_infinite-dimensional_2020} that such a solution~$u$ exists and is unique.

\medskip

An \emph{equilibrium}  of~\eqref{eq:SIS2} is a function~$g  \in \Deltac$
such    that~$F(g)     =    \zero$    (in    $\cl$).     According    to
\cite{delmas_infinite-dimensional_2020},   there    exists   a   maximal
equilibrium~$\mathfrak{g}$, \textit{i.e.}, an  equilibrium such that all
other   equilibria~$h\in  \Deltac$   are  dominated   by~$\mathfrak{g}$:
$h \leq  \mathfrak{g}$. The \emph{reproduction  number}~$R_0$ associated
to the SIS model given by~\eqref{eq:SIS2}  is the spectral radius of the
next-generation operator, so that using  the definition of the effective
reproduction       number~\eqref{eq:def-R_e},~\eqref{eq:def-R0}      and
\eqref{eq:rhoT=rhoT}, this amounts to:
\begin{equation}\label{eq:def-R0-2}
  R_0= \rho (\Tinf_{k/\gamma})=R_0[k/\gamma]= R_e[k/\gamma](\un).
\end{equation}
If~$R_0\leq 1$  (sub-critical and  critical regime),  then~$u_t$ converges
pointwise   to~$\zero$  when~$t\to\infty$.    In  particular,   the  maximal
equilibrium~$\mathfrak{g}$  is  equal   to~$0$  everywhere.   If~$R_0>1$
(super-critical regime),  then the null  function is still  an equilibrium
but   different    from   the   maximal    equilibrium   $\mathfrak{g}$,
as~$\int_\Omega \mathfrak{g} \, \mathrm{d}\mu > 0$.

\subsection{Vaccination strategies in the SIS model}\label{sec:vacc}

A \emph{vaccination strategy}~$\eta$ of a vaccine with perfect
efficiency is an element of~$\Deltac$, where~$\eta(x)$ represents the
proportion of \emph{\textbf{non-vaccinated}} individuals with
feature~$x$. Notice that~$\eta\, \mathrm{d} \mu$ corresponds in a
sense to the effective population.

Recall the definition of the kernel~$f\kk g$
from~\eqref{eq:def-fkg}. For~$\eta \in \Deltac$, the
kernels~$k\eta/\gamma$ and~$k\eta$ have finite norm
$\norm{\cdot}_{\infty , q}$ under Assumption \ref{hyp:k-g}, so we can
consider the bounded positive operators~$\Tinf_{k \eta / \gamma}$
and~$\Tinf_{k\eta}$ on~$\mathscr{L}^\infty$. According to
\cite[Section~5.3.]{delmas_infinite-dimensional_2020}, the SIS
equation with vaccination strategy~$\eta$ is given by~\eqref{eq:SIS2},
where~$F$ is replaced by~$F_\eta$ defined by:
\begin{equation}
  \label{eq:vec-field-vaccin}
  F_\eta(g) = (\un - g) \Tinf_{k\eta}(g) - \gamma g.
\end{equation}
We denote by~$u^\eta=(u^\eta_t, t\geq 0)$ the corresponding solution
with initial condition~$u_0^\eta \in \Deltac$.  We recall
that~$u_t^\eta(x)$ represents the probability for an non-vaccinated
individual of feature~$x$ to be infected at time $t$.  Since the
effective reproduction number is the spectral radius
of~$\Tinf_{k\eta/\gamma}$, we recover~\eqref{eq:def-R_e} with
$\kk=k/\gamma$
as~$ \rho(\Tinf_{k\eta/\gamma})=\rho(T_{k\eta/\gamma} )=
R_e[k/\gamma](\eta)$.  We denote by~$\mathfrak{g}_\eta$ the
corresponding maximal equilibrium (so that
$\mathfrak{g}=\mathfrak{g}_{\un}$). In particular, we have:
\begin{equation}
  \label{eq:F(g)=0}
  F_\eta(\mathfrak{g}_\eta)=\zero \quad \text{(in $\cl$)}.
\end{equation}
We will denote by~$\I $ the \emph{fraction of infected individuals at
  equilibrium}. Since the probability for an individual with
feature~$x$ to be infected in the stationary regime
is~$\mathfrak{g}_\eta(x) \, \eta(x)$, this fraction is given by the
following formula:
\begin{equation}
  \label{eq:asymptotic_number_endemic}
  \I (\eta)=\int_\Omega \mathfrak{g}_\eta
  \, \eta\, \mathrm{d}\mu=\int_\Omega \mathfrak{g}_\eta(x)
  \, \eta(x)\, \mu(\mathrm{d}x).
\end{equation}
We deduce from~\eqref{eq:vec-field-vaccin} and~\eqref{eq:F(g)=0}
that~$\mathfrak{g}_\eta\eta=0$ $\mu$-almost surely is equivalent to~$\mathfrak{g}_\eta=0$.
Applying the results of~\cite{delmas_infinite-dimensional_2020} to the kernel~$k \eta$, we
deduce that:
\begin{equation}
  \label{eq:gh>0}
  \I (\eta)>0 \,\Longleftrightarrow\, R_e[k/\gamma](\eta)>1.
\end{equation}

\section{General properties of the functions
  \texorpdfstring{$R_e$}{Re}
  and 
  \texorpdfstring{$\I$}{I}}%
\label{sec:properties-loss}

As mentioned in the introduction, see \cite{ddz-theory-optim}, we shall see the
functions  $R_e$  and $\I$  defined  on  $\Deltac\subset\cl$ and  taking
values in $\R_+$  as loss functions, and check they  are well defined on
$\Deltad\subset L^\infty$,  see Definition~\ref{defi:well-def},  and then
non-decreasing and continuous on $\Deltad$.

\subsection{The effective reproduction number~\texorpdfstring{$R_e$}{Re}}

We  consider  the  kernel  model $[(\Omega,  \cf\!,  \mu),  \kk]$  under
Assumption~\ref{hyp:k},  so  that~$\mu$  is a  non-zero  $\sigma$-finite
measure  and   $\kk$  is  a   kernel  on~$\Omega$  with   finite  double
norm.  Recall  the  effective  reproduction  number  function~$R_e[\kk]$
defined                            on                           $\Deltac$
by~\eqref{eq:def-R_e}:~$R_e[\kk](\eta)=\rho(T_\kk   M_\eta)$,  and   the
reproduction  number~$R_0[\kk]=\rho(T_\kk)$. When  there is  no risk  of
confusion  on   the  kernel~$\kk$,  we  simply   write  $R_e$  and~$R_0$
for~$R_e[\kk]$ and~$R_0[\kk]$.

\begin{proposition}[Basic properties of $R_e$]\label{prop:R_e}
  Suppose Assumption \ref{hyp:k} holds.  The
  function~$R_e=R_e[\kk]$ satisfies the following properties:
  \begin{enumerate}[(i)]
  \item\label{prop:a.s.+increase-Re}%
    The function $R_e$  is well  defined and non-decreasing  on $\Deltad$
    endowed with the weak-* topology.
   \item\label{prop:min_Re}%
    $R_e(\zero) = 0$ and~$R_e(\un) = R_0$.
   \item\label{prop:normal}%
    $R_e(\lambda \eta) = \lambda R_e(\eta)$ for all~$\eta\in \Deltad$
    and~$\lambda \in [0,1]$. 
  \end{enumerate}
\end{proposition}

\begin{proof}
  If~$\eta_1=\eta_2$~$\mu$-almost surely, then we have
  that~$T_{\kk\eta_1} = T_{\kk\eta_2}$, and
  thus~$R_e(\eta_1)=R_e(\eta_2)$.  If~$\eta_1 \leq \eta_2$
  $\mu$-almost everywhere, then the
  operator~$T_{\kk \eta_2} - T_{\kk \eta_1 }$ is positive.  According
  to~\eqref{eq:spec_rad_croissant}, we get that
  $\rho(T_{\kk \eta_1 }) \leq \rho(T_{k \eta_2 })$.  This concludes
  the proof of Point~\ref{prop:a.s.+increase-Re}.
  Point~\ref{prop:min_Re} is a direct consequence of the definition
  of~$R_e$.  Since for any fixed~$\lambda\in\R_+$ and any operator~$T$
  on $L^p$, the norm of~$\lambda T$ is equal
  to~$\lambda\norm{T}_{L^p}$, Point~\ref{prop:normal} is clear.
\end{proof}

Similarly, we get that the function $\spec[\kk]$ defined on $\Deltac$ is
well defined on $\Deltad$. 
We generalize a continuity property on the spectral radius originally stated
in~\cite{delmas_infinite-dimensional_2020} by weakening the topology.

\begin{theorem}[Continuity of~$\grR$ and~$\grS$]\label{th:continuity-R}
  Suppose Assumption \ref{hyp:k} holds. Then, the functions~$\spec[\kk]$
  and~$R_e[\kk]$ are  continuous functions from $\Deltad$  (endowed with
  the weak-* topology) respectively to~$\ck$ (endowed with the Hausdorff
  distance) and to~$\R_+$ (endowed with the usual Euclidean distance).
\end{theorem}

Let us remark the proof holds even if~$\kk$ takes negative values.

\begin{proof}
  Let~$B$ denote the unit ball in $L^p$, with $p\in (1, +\infty )$
  from Assumption \ref{hyp:k}. Since the operator~$T_\kk$ on $L^p$ is
  compact, the set~$T_{\kk}(B)$ is relatively compact.  For
  all~$\eta\in \Deltad$, set~$\eta B=\{\eta g\, \colon\, g\in
  B\}$. As~$\eta B\subset B$, we deduce
  that~$T_{\kk \eta}(B)= T_\kk (\eta B) \subset T_\kk (B)$. This
  implies that the
  family~$\{T_{\kk \eta} \, \colon\, \eta \in \Deltad\}$ is
  collectively compact.

  Let~$(\eta_n , \, n \in \N)$ be a sequence in~$\Deltad$ weak-*
  converging to some~$\eta \in \Deltad$. Let~$g \in L^p$. The weak-*
  convergence of~$\eta_n$ to~$\eta$ implies
  that~$(T_{\kk \eta_n}(g), \, n \in \N)$ converges~$\mu$-almost
  surely to~$T_{\kk \eta} (g)$. Consider the function $K$ defined on
  $\Omega$ by:
  \[
    K(x)=\left(\int _\Omega \kk(x,y)^q \, \mu(\mathrm{d} y) \right)^{1/q},
  \]
  which belongs to~$L^p$, thanks to~\eqref{eq:Lp-integ-cond}. Since for
  all~$x$,
  \[
    \abs{T_{\kk \eta_n}(g)(x)}
    \leq T_{\kk}( \abs {\eta_n g})(x)
    \leq K(x) \, \| \eta_n g\|_p \leq K(x)\, \norm{g}_p ,
  \]
  we deduce, by dominated convergence, that the convergence holds also in~$L^p$:
  \begin{equation}\label{eq:cv-Tkn}
    \lim_{n\rightarrow \infty } \norm{ T_{\kk \eta_n}(g) - T_{\kk \eta}(g)}_p=0,
  \end{equation}
  so  that $T_{\kk\eta_n}$  converges strongly  to $T_{\kk\eta}$.  Using
  Lemma~\ref{lem:prop-spec-mult}
  (with~$T_n=T_{\kk \eta_n}$ and~$T=T_{\kk \eta}$)  on the continuity of
  the                  spectrum,                 we                  get
  that~$\lim_{n\rightarrow\infty                                       }
  \spec[\kk](\eta_n)=\spec[\kk](\eta)$. The function~$\spec[\kk]$
  is  thus  weak-*  sequentially
  continuous, and, thanks to Lemma \ref{lem:D-compact}, it is continuous
  from~$\Deltad$  endowed  with  the   weak-*  topology  to  the  metric
  space~$\ck$  endowed  with  the  Hausdorff  distance.  The  continuity
  of~$R_e[\kk]$ then  follows from its  definition~\eqref{eq:def-rho} as
  the   composition    of   the    continuous   functions~$\mathrm{rad}$
  and~$\spec[\kk]$.
\end{proof}

We now give a stability property of the spectrum and
spectral radius with respect to the kernel~$\kk$.

\begin{proposition}[Stability of~$\grR$ and~$\grS$]\label{prop:Re-stab}
  Let $\mu$ be a $\sigma$-finite non-zero  measure on the measurable space
  $(\Omega,  \cf)$.  Let~$p\in  (1, +\infty  )$. Let~$(\kk_n,  n\in \N)$
  and~$\kk$   be   kernels   on~$\Omega$  with   finite   double   norms
  on~$L^p$.  If~$\lim_{n\rightarrow\infty} \norm{\kk_n  - \kk}_{p,q}=0$,
  then we have:
  \begin{equation}
    \label{eq:Re-stab}
    \lim_{n\rightarrow\infty }\, \sup_{\eta\in \Deltac} \Big|R_e[\kk_n](\eta) -
    R_e[\kk](\eta)\Big|=0
    \quad\text{and}\quad
    \lim_{n\rightarrow\infty }\, \sup_{\eta\in \Deltac} d_\mathrm{H}\Big (
    \spec[\kk_n](\eta), \spec [\kk](\eta)\Big)=0.
  \end{equation}
\end{proposition}

\begin{proof}
  Notice the suprema in~\eqref{eq:Re-stab} can also be taken over
  $\Deltad$ as $R_e$ and $\spec$ defined on $\Deltac$ are well defined
  on $\Deltad$.  Let us first prove that, if $(\eta_n, n\in \N)$ is a
  sequence in~$\Deltad$ which weak-* converges to~$\eta\in \Deltad$,
  then $ \spec[\kk_n](\eta_n)$ converges to $\spec[\kk](\eta)$ in
  Haussdorff distance.
  
  All the operators in~$\ca=\{T_\kk\} \cup \{T_{\kk_n}\, \colon\, n\in
  \N\}$
  are compact, and we deduce from~\eqref{eq:double-norm-norm} that:
  \[
    \lim_{n\rightarrow\infty}\norm{T_{\kk_n} -T_\kk}_{L^p}=0.
  \]
  Therefore $\ca$ is a compact set in~$\cll(L^p)$. Since the elements of
  $\ca$    are    compact    operators,    we    get    by~\cite[Theorem
  2.4]{anselone_palmer_collectively_compact}, that $\ca$ is collectively
  compact.
  Since $\{M_\eta\, \colon\, \eta\in \Deltad\}$ is a bounded set
  in~$\cll(L^p)$, we deduce from \cite[Proposition~4.2(2)]{anselone},
  that the
  family~$\ca'=\{T'M_\eta\, \colon, T'\in \ca \text{ and } \eta\in
  \Deltad\}$ is collectively compact.  A fortiori the
  sequence~$(T_n=T_{\kk_n \eta_n} =T_{\kk_n}M_{\eta_n}, n\in \N)$ of
  elements of~$\ca'$ is collectively compact, and
  $T=T_{\kk\eta}=T_\kk M_\eta$ is compact.

  Let~$g\in L^p$. We have:
  \[
    \norm{T_n(g) - T(g)}_p
    \leq \norm{T_{\kk_n}-T_\kk}_{L^p} \, \norm{ g}_p + \norm{T_{\kk\eta_n}(g) -
    T_{\kk\eta}(g)}_p.
  \]
  Using~$\lim_{n\rightarrow\infty    }\norm{T_{\kk_n}   -T_\kk}_{L^p}=0$
  and~\eqref{eq:cv-Tkn},                      we                     get
  that~$\lim_{n\rightarrow\infty    }     \norm{T_n(g)    -    T(g)}_p$,
  thus~$(T_n,   n\in  \N)$   converges  strongly   to  $T$.   Thanks  to
  Lemma~\ref{lem:prop-spec-mult},    we    deduce
  that~$\lim_{n\rightarrow  \infty   }  \spec(T_n)=\spec(T)$,   that  is
  $\lim_{n\rightarrow\infty            }           \spec[\kk_n](\eta_n)=
  \spec[\kk](\eta)$. \medskip

  Then,                              as                              the
  function~$\eta  \mapsto d_\mathrm{H}\Big  ( \spec[\kk_n](\eta),  \spec
  [\kk](\eta)\Big)$  is   weak-*  continuous   on  the   weak-*  compact
  set~$\Deltad$, thanks to Theorem~\ref{th:continuity-R}, it reaches its
  maximum  say at~$\eta_n\in  \Deltad$  for~$n\in  \N$. As~$\Deltad$  is
  weak-* sequentially   compact,  consider  a  sub-sequence which  weak-*
  converges to a limit say~$\eta$. Since
  \begin{multline*}
    \sup_{\eta\in \Deltad} d_\mathrm{H}\Big (
    \spec[\kk_n](\eta), \spec [\kk](\eta)\Big)\\
    \begin{aligned}
 & = d_\mathrm{H}\Big (
 \spec[\kk_n](\eta_n), \spec [\kk](\eta_n)\Big) \\
 & \leq d_\mathrm{H}\Big (
 \spec[\kk_n](\eta_n), \spec [\kk](\eta)\Big)
 + d_\mathrm{H}\Big (
 \spec[\kk](\eta_n), \spec [\kk](\eta)\Big),
    \end{aligned}
  \end{multline*}
  using the weak-* continuity of~$\spec[\kk]$,  we deduce that along this
  sub-sequence the  right hand  side converges to  0. Since  this result
  holds  for  any  converging  sub-sequence,  we  get  the  second  part
  of~\eqref{eq:Re-stab}.   The  first   part  then   follows  from   the
  definition~\eqref{eq:def-rho}  of  $R_e$  as a  composition,  and  the
  Lipschitz continuity of the function~$\mathrm{rad}$.
\end{proof}

\subsection{The asymptotic proportion of infected individuals \texorpdfstring{$\I $}{}}
   We consider  the SIS  model $[(\Omega, \cf,  \mu), k,  \gamma]$ under
   Assumption~\ref{hyp:k-g}.                                      Recall
   from~\eqref{eq:asymptotic_number_endemic}    that   the    asymptotic
   proportion  of infected  individuals~$\I $  is given  on $\Deltac$  by
   $\I  (\eta)=\int_\Omega \mathfrak{g}_\eta  \, \eta\,  \mathrm{d}\mu$,
   where $\mathfrak{g}_\eta$ is the  maximal solution in~$\Deltac$ of the
   equation  $F_\eta(h)  =  0$.   We first  recall  \cite[Lemma~5.3  and
   Proposition~5.5]{ddz-hit} on the properties and characterization of
   the maximal equilibrium $\mathfrak{g}=\mathfrak{g}_\un$.

\begin{lemma}[Properties of the maximal equilibrium]
  \label{lem:Fh>0}
  Suppose Assumption \ref{hyp:k-g} holds.
  \begin{enumerate}[(i)]
    \item  Let $\eta, g \in \Deltac$. If $F_\eta(g) \geq 0$, then we have $g\leq
  \mathfrak{g}_\eta$ (in $\cl$).
\item\label{cor:h=g} For any~$h\in\Deltac$, we have~$h=\mathfrak{g}$
  (in $\cl$) if and only if~$F(h) = 0$  (in $\cl$)
    and~$R_e(\un-h)\leq 1$.
  \item\label{cor:g>0}         If~$R_0>1$        (or         equivalently
    $\mathfrak{g}\neq \zero$  in $\cl$), then we have~$R_e(\un-\mathfrak{g}) = 1$.
  \end{enumerate}
\end{lemma}

We may now state the main properties of the function $\I$.

\begin{proposition}[Basic properties of $\I $]\label{prop:I}
  Suppose that Assumption~\ref{hyp:k-g} holds and write $R_e$ for
  $R_e[k/\gamma]$.  The
  function~$\I $ has the following properties:
  \begin{enumerate}[(i)]
    \item\label{prop:a.s.+increase-I}%
    The function $\I$  is well  defined and non-decreasing  on $\Deltad$
    endowed with the weak-* topology.
  \item\label{prop:min-I} For $\eta\in \Deltad$, we have $\I (\eta)=0$ if
    and only if $R_e(\eta) \leq 1$.
   \item\label{prop:hom-I} $\I (\lambda \eta) \leq \lambda \I (\eta)$ for
    all~$\eta\in \Deltad$ and~$\lambda \in
    [0,1]$.
  \end{enumerate}
\end{proposition}

\begin{proof}
  If $\eta_1=\eta_2$ $\mu$-almost surely, then the operators
  $\Tinf_{k\eta_1}$ and $\Tinf_{k\eta_2}$ are equal.  Thus, the
  equilibria $\mathfrak{g}_{\eta_1}$ and $\mathfrak{g}_{\eta_2}$ are
  also equal, which in turns implies that $\I (\eta_1)=\I (\eta_2)$.
  To prove the monotonicity, consider $\eta_1, \eta_2\in \Deltad$ such
  that a.s.\ $\eta_1\leq \eta_2$. This gives
  $\Tinf_{k\eta_1} \leq \Tinf_{k\eta_2}$. We deduce that
  $F_{\eta_1}(g)\leq F_{\eta_2}(g)$ in $\cl$ for all
  $g\in\Deltac\subset \cl$.  In particular, taking
  $g=\mathfrak{g}_{\eta_1}$ and using~\eqref{eq:F(g)=0}, we get
  $F_{\eta_2}(\mathfrak{g}_{\eta_1})\geq 0$.  By Lemma~\ref{lem:Fh>0}
  this implies $\mathfrak{g}_{\eta_1}\leq \mathfrak{g}_{\eta_2}$. To
  sum up, we get:
   \begin{equation}
    \label{eq:mon-gh}
    \eta_1 \leq \eta_2 \quad \text{in
     $L^\infty$}\quad \Longrightarrow \quad \mathfrak{g}_{\eta_1}\leq
    \mathfrak{g}_{\eta_2} \quad\text{in $\cl$.}
  \end{equation}
  This readily implies that $ \I (\eta_1) = \int_\Omega
  \mathfrak{g}_{\eta_1}\, \eta_1 \, \mathrm{d}\mu \leq \int_\Omega
  \mathfrak{g}_{\eta_2}\, \eta_2 \, \mathrm{d}\mu = \I (\eta_2)$.
This gives   Point~\ref{prop:a.s.+increase-I}.

  \medskip

  Point~\ref{prop:min-I} is already stated in Equation~\eqref{eq:gh>0}.
  We now consider Point~\ref{prop:hom-I}. Since $\lambda\in [0, 1]$, we deduce
  from~\eqref{eq:mon-gh} that $\mathfrak{g}_{\lambda \eta} \leq
  \mathfrak{g}_{\eta}$. This implies that $ \I (\lambda \eta) = \int_\Omega
  \mathfrak{g}_{\lambda \eta}\, \lambda \eta \, \mathrm{d}\mu \leq \lambda \int_\Omega
  \mathfrak{g}_{\eta}\, \eta \, \mathrm{d}\mu = \lambda \I (\eta)$.
\end{proof}

The proof of the following continuity results are both postponed to
Section~\ref{sec:technical_proofs}.

\begin{theorem}[Continuity of $\I$]\label{th:continuity-I}
  Suppose that Assumption~\ref{hyp:k-g} holds. The function~$\I $ defined on~$\Deltad$ is
  continuous with respect to the weak-* topology.
\end{theorem}

We write $\I[k, \gamma]$ for $\I$ to stress the dependence on the
parameters $k, \gamma$ of the SIS model.

\begin{proposition}[Stability of $\I$]\label{prop:I-stab}
  Let $((k_n, \gamma_n), n\in \N)$ and $(k,\gamma)$ be a sequence of kernels and functions
  satisfying Assumption~\ref{hyp:k-g}. Assume furthermore that there exists $p'\in (1,
  +\infty )$ such that $\kk=\gamma^{-1} k$ and $(\kk_n=\gamma^{-1}_n k_n, n\in \N)$ have
  finite double norm in $L^{p'}$ and that $\lim_{n\rightarrow\infty } \norm{\kk_n
  -\kk}_{p',q'}=0$. Then we have:
  \begin{equation}\label{eq:I-stab}
    \lim_{n\rightarrow\infty }\, \sup_{\eta\in \Deltac} \Big|\I[k_n, \gamma_n](\eta) -
    \I[k,\gamma](\eta)\Big| = 0.
  \end{equation}
\end{proposition}

Let us stress that Assumption~\ref{hyp:k-g} on $k $ and $\gamma$ implies
that $k \gamma^{-1}$ has a finite double norm. In the proposition above,
it is also assumed that $\gamma^{-1} k$ has a finite double norm. Notice
those two conditions coincide when $\essinf_\Omega \gamma$ is positive.

\section{Other properties of the functions \texorpdfstring{$R_e$}{Re}
  and  
\texorpdfstring{$\I$}{I}}
\label{sec:other-reg}
In the companion paper \cite{ddz-theory-optim}  we consider the optimization of
the  protection  of  the  population,   which  can  be  written  as  the
bi-objective minimization problem $\min(C(\eta),\loss(\eta))$, where $C$
and $\loss$ stand  respectively for the cost and the  loss incurred when
following the  vaccination strategy~$\eta$. In  our setting the  loss is
given either by the effective  reproduction number $R_e$ or the fraction
of  infected individuals  at equilibrium  $\I$.   (To fix  the ideas,  a
natural cost $C$, when the measure  $\mu$ is finite, is the uniform cost
$  \costu (\eta)=\int_\Omega  (\un-\eta)\, \mathrm{d}\mu$  corresponding
intuitively  to the  number of  doses used  in the  vaccination strategy
$\eta$,  as  we  recall  that  $1-\eta(x)$  is  the  proportion  of  the
vaccinated  population  with  given   feature  $x$. Notice the cost
$\costu$ is well defined on $\Deltad\subset L^\infty $.)   The  bi-objective
minimization  problem  is  then  studied under  some  of  the  following
hypothesis on the loss.
Recall that  $\Deltad$ is  endowed with 
  the weak-* topology.

\begin{hyp}[On the loss]
  \label{hyp:loss-topo}
Let  $\loss$ be  a loss function from $\Deltad$ to $\R$.
  \begin{enumerate}[(i)]
  \item \textbf{Monotony.} \label{hyp:loss+cost}
    The  function $\loss$ is non-decreasing
  continuous  with $
    \loss(\zero)=0 $ and $   \max_\Deltad \, \loss>0$.
  \item \textbf{Minima.} \label{hyp:loss}
    Any local minimum of the  function $\loss$ is a  global minimum.

  \item \textbf{Maxima.} \label{hyp:loss**}
    Any local maximum of the  function $\loss$ is a  global maximum.
  \end{enumerate}
\end{hyp}

Notice the loss functions $R_e$ and $\I$ satisfy clearly
Assumption~\ref{hyp:loss-topo}~\ref{hyp:loss+cost} provided they are not
trivially equal to zero. Thanks to the next lemma, they also satisfy
Assumption~\ref{hyp:loss-topo}~\ref{hyp:loss} as they are
sub-homogeneous.

\begin{lemma}\label{lem:c-dec+L-hom}
  Let  $  \loss$ be  a  non-negative  and non-decreasing  loss  function
  defined   on   $\Deltad$.    If  it   is
  sub-homogeneous,  then  Assumption  \ref{hyp:loss-topo}~\ref{hyp:loss}
  holds.
\end{lemma}

\begin{proof}
  Let $\eta\in\Deltad$. If $\loss$ has
  a local minimum at $\eta$, then for $\varepsilon>0$ small enough $\loss(\eta) \leq
  \loss((1-\varepsilon) \eta) \leq (1-\varepsilon) \loss(\eta)$, so $\loss(\eta) = 0$ and
  $\eta$ is a global minimum of $\loss$.
\end{proof}

We prove in this section that under some irreducibility condition on the
kernel                that                $R_e$                satisfies
Assumption~\ref{hyp:loss-topo}~\ref{hyp:loss**}.  The situation is a bit
more     complicated     for     the     loss     $\I$,     for     which
Assumption~\ref{hyp:loss-topo}~\ref{hyp:loss**} does  not hold. However,
$\I$  satisfies  a  weakened  version,  see
Assumption~\ref{hyp:loss=I}~\ref{item:hyp:loss=I} 
below.
The reducible  case is more delicate  and it is studied  in more details
in~\cite[Section~5]{ddz-cordon}  for the  loss function  $\loss=R_e$; in
particular Assumption~\ref{hyp:loss-topo}~\ref{hyp:loss**}  may not hold
in this case.

In Section~\ref{sec:atom} we consider some irreducibility property for a
kernel and its implications for the SIS model, see
also~\cite{dlz-atom,dlz-equilibre} for further results in this
direction. In Section~\ref{sec:kernel-hyp},  we provide some
irreducibility conditions in the kernel model so that the loss function $R_e$
satisfies Assumption~\ref{hyp:loss-topo}, see Lemma~\ref{lem:L**}. 
Section~\ref{sec:SIS-hyp} provide similar results for the loss $\I$ in
the SIS model, see Lemma~\ref{lem:L**=I}.

\subsection{Irreducible, quasi-irreducible and monatomic kernels}\label{sec:atom}

We  follow the  presentation in~\cite[Section~5]{ddz-Re}  on the  atomic
decomposition of  positive compact operator  and Remark 5.2  therein for
the  particular case  of  integral operators,  see  also the  references
therein for  further results.   Let $(\Omega, \cf,  \mu)$ be  a measured
space with $\mu$ a non-zero $\sigma$-finite measure.  For $A, B\in \cf$,
we write  $A\subset B$ a.e.\ if  $\mu(B^c \cap A)=0$ and  $A=B$ a.e.\ if
$A\subset B$  a.e.\ and $B\subset  A$ a.e..   Let $\cg\subset \cf$  be a
$\sigma$-field. A  set $A$ is  an \emph{atom} of  $\mu$ in $\cg$  if $A$
belongs to  $\cg$, and  for all  $B\subset A$ with  $B\in \cg$,  we have
either $B=\emptyset$ a.e.\ or $B=A$ a.e.. Notice that the atoms are
defined up to an a.e.\ equivalence. 

Let $\kk$  be a
kernel on  $\Omega$ with a  finite double norm.  For   $A,  B\in  \cf$,
$x\in            \Omega$,           we            simply           write
$\kk(x,A)=\int_{        A}         \kk(x,y)\,        \mu(\rd        y)$,
$ \kk(B,x)=\int_{ B} \kk(z,x)\, \mu(\rd z)$ and:
\[
  \kk(B, A)=
  \int_{B \times A} \kk(z,y)\, \mu(\rd z) \mu(\rd y) \in [0, +\infty ].
\]
A   set  $A\in   \cf$  is   called  \emph{$\kk$-invariant},   or  simply
\emph{invariant}  when there  is no  ambiguity on  the kernel  $\kk$, if
$\kk(A^c,  A)=0$.   In  the  epidemiological setting,  the  set  $A$  is
invariant if the  sub-population $A$ does not  infect the sub-population
$A^c$.  The kernel $\kk$  is \emph{irreducible} (or \emph{connected}) if
any invariant set $A$ is such  that $\mu(A)=0$ or $\mu(A^c)=0$. If $\kk$
is irreducible, then either $R_0[\kk]>0$  or $k\equiv 0$ and $\Omega$ is
an  atom  of $\mu$  in  $\cf$  (degenerate  case). A  simple  sufficient
condition for irreducibility is for the kernel to be a.e.\ positive.

\medskip

Let $\ca$ be the set of $\kk$-invariant sets. Let us stress that the set
of $\kk$-invariant sets depends only on the support of the kernel $\kk$.
In particular in the SIS model, with $\kk=k/\gamma$, the $\kk$-invariant
sets and the  $k$-invariant sets coincide.  Notice that  $\ca$ is stable
by countable unions and countable intersections.  Let $\cfi=\sigma(\ca)$
be  the $\sigma$-field  generated by  $\ca$.  Then,  the operator  $\kk$
restricted to an atom of $\mu$  in $\cfi$ is irreducible.  We shall only
consider non degenerate atoms, and say  the atom (of $\mu$ in $\cfi$) is
non-zero if the restriction of the kernel $\kk$ to this atom is non-zero
(and thus the spectral radius  of the corresponding integral operator is
positive). 
We say  the kernel $\kk$  is \emph{monatomic}  if there exists  a unique
non-zero atom, say $\oa$, and  the kernel is \emph{quasi-irreducible} if
it is monatomic  and $\kk\equiv 0$ outside $\oa\times  \oa$, where $\oa$
is its  non-zero atom.  Notice  that: (i)  if $\kk$ is  irreducible with
$R_0[\kk]>0$, then  $\kk$ is monatomic with  non-zero atom $\oa=\Omega$;
(ii)  if  $\kk$ is  monatomic,  then  $R_0[\kk]>0$ by  definition.   The
quasi-irreducible  property is  the usual  extension of  the irreducible
property in the setting of symmetric kernels; and the monatomic property
is the natural generalization to non-symmetric kernels.
  
\medskip

According to~\cite[Lemma~5.3]{ddz-Re},  we get  that if a  kernel $\kk$,
with  finite double  norm, is  monatomic  with non-zero  atom $\oa$  and
$\eta\in      \Deltac$,      then,      with~$\ka=\ind{\oa}\kk\ind{\oa}$
and~$\bm{\kk}_\oaa$ (resp.\  $\bm{\eta}_\oaa$) the restriction  of $\kk$
(resp.\ $\eta\in \Deltac$) to $\oa$:
\begin{equation}\label{eq:decomp}
  R_e[\kk](\eta)=  R_e[\ka ](\eta)
  =  R_e[\ka ](\eta\ind{\oa})
  = R_e[\bm{\kk}_\oaa](\bm{\eta}_\oaa).
\end{equation}
 
\begin{remark}[Epidemiological interpretation]\label{rem:lien-epidemie}
   When the  kernel $\kk=k/\gamma$ for   the    SIS   model
  $[(\Omega, \cf, \mu), k,  \gamma]$ is  monatomic, with
   non-zero atom $\oa$,  then the population with trait in  $\oa$ can infect
   itself. It may also infect another part of the population, say with
   trait in $\oi$, but:
   \begin{itemize}
   \item the infection cannot be sustained at all in $\oi$: $\kk$ is
     quasi-nilpotent on $\oi$;
     \item the population with trait in $\oi$ does not infect
       back the non-zero atom $\oa$.  
     \end{itemize}
      If  furthermore  $R_0>1$, then  the  set  $\oa\cup\oi$
   corresponds  to  the  support  of  the  maximal  endemic
   equilibrium.
\end{remark}

In the monatomic case, the  non-zero equilibrium, if it exists, is unique. 
This  result is a direct consequence of   Lemma~4.1~(iii) and Corollary~4.11 in
   \cite{dlz-equilibre}. 

\begin{lemma}[Equilibrium in the monatomic case]\label{lem:g>0=oa+oi}
  Assume    Assumption~\ref{hyp:k-g}   holds    for   the    SIS   model
  $\param=[(\Omega, \cf, \mu), k,  \gamma]$ and that $R_0[\kk]>1$ (super-critical regime) with $\kk=k/\gamma$. If $k$ (and $\kk$) is monatomic,
  with  non-zero atom  say $\oa$,  then there  exists a  unique non-zero
  equilibrium, say $\mathfrak{g}$, and its support is the smallest
  invariant set containing $\oa$, that is, the set  $\{\mathfrak{g}>0\}$ is
  invariant and if $A$ is invariant and $\oa \subset A$, then
  a.s. $\{\mathfrak{g}>0\} \subset A$. 
\end{lemma}

\subsection{The kernel model}
\label{sec:kernel-hyp}
We now check Assumptions~\ref{hyp:loss-topo}~\ref{hyp:loss}-\ref{hyp:loss**} for the
loss $\loss=R_e$.

\begin{lemma}[Extrema of $R_e$]
  \label{lem:L**}
  Consider  the kernel  model $\param=[(\Omega,  \cf, \mu),  \kk]$ under
  Assumption~\ref{hyp:k}, and simply write $R_e$ for the loss function $\loss=R_e[\kk]$.

  \begin{enumerate}[(i)]
  \item   Assumption~\ref{hyp:loss-topo}~\ref{hyp:loss+cost}  holds   if
    $R_0>0$.
\item   Assumption~\ref{hyp:loss-topo}~\ref{hyp:loss}   holds.

\item\label{it:Re-mono-H3}  If $\kk$ is monatomic with atom $\oa$, then $R_0>0$ and Assumption
  \ref{hyp:loss-topo}~\ref{hyp:loss**} holds. Furthermore, 
  $\eta\in  \Deltad$  is a  global  maximum  of  $R_e$  if and  only  if
  $\eta\geq \ind{\oa}$ (in $L^\infty $).
\end{enumerate}
\end{lemma}

\begin{proof}
  Since the function  $R_e$ is homogeneous, see
  Proposition~\ref{prop:R_e}, we deduce from Lemma~\ref{lem:c-dec+L-hom}
  that Assumption~\ref{hyp:loss-topo}~\ref{hyp:loss} holds. Using
  Theorem~\ref{th:continuity-R}, for the continuity, 
  Proposition~\ref{prop:R_e}, for the monotonicity of the function
  $R_e$,  and the fact that  $R_0>0$,
  the   hypotheses on the loss in Assumption~\ref{hyp:loss-topo}~\ref{hyp:loss+cost} hold. 
   \medskip

   We  now prove  Point~\ref{it:Re-mono-H3}.  We first  assume that  the
   kernel  $\kk$ is  irreducible with  $R_0>0$. In  particular, we  have
   a.e.\  that  $\kk(\Omega, y)>0$.   Let  $\eta\in\Deltad$  be a  local
   maximum; we want to show that it is also a global maximum.

   Suppose  first  that $\inf\eta  >0$.  Then  $\kk\eta$ is  irreducible
   non-zero   with   finite   double  norm.    According   \cite[Theorem
   V.6.6]{schaefer_banach_1974} and  since $T_{\kk \eta}=  T_\kk M_\eta$
   is  compact,   the  eigenspace   of  $T_{\kk  \eta}$   associated  to
   $R_e(\eta)$ is  one-dimensional and it  is spanned by a  vector $\vd$
   such  that  $\vd>0$  a.e.,  and the  corresponding  left  eigenvector
   associated  to  $R_e(\eta)$,  say  $\vg$, can  be  chosen  such  that
   $\langle  \vg,   \vd\rangle=1$  and   a.e.\  $\vg>0$.   According  to
   \cite[Theorem       2.6]{EffectivePertuBenoit},      applied       to
   $L_0=T_{\kk\eta}$ and $L =  T_{\kk(\eta + \varepsilon(\un-\eta))}$ with
   $\varepsilon\in     (0,     1)$,      we     have,     using     that
   $\norm{L_0        -        L}_{L^p}=O(\varepsilon)$        thanks        to
   \eqref{eq:double-norm-norm}:
  \[
    R_e(  \eta+\varepsilon(\un-\eta)) =  R_e(\eta) +  \varepsilon \langle
    \vg,T_{\kk(\un-\eta)} \vd  \rangle + O(\varepsilon^2).
  \]
  Since $R_e$ has a local maximum at $\eta$, the first order term on the
  right hand side vanishes, so $\vg(x)\kk(x,y)(\un-\eta(y))\vd(y) = 0$ for
  $\mu$  almost every   $x$ and  $y$. Since  $\vg$ and  $\vd$ are  positive
  a.e.\    and   $\kk$    is   irreducible,   we    get   that
  $\kk(\Omega, y)(\un-  \eta(y))=0$ a.e.\  and thus a.e.\ $\eta(y)  = 1$.
    Therefore  $\eta = \un$, which is a  global maximum for
  $R_e$.

  Finally, suppose  that $\inf\eta =  0$. Let $G$  be an open  subset of
  $\Deltad$  on which  $R_e\leq  R_e(\eta)$ and  with  $\eta\in G$.  For
  $\varepsilon>0$        small        enough,        the        strategy
  $\eta_\varepsilon  = \eta  +\varepsilon(\un-\eta)$  belongs  to $G$  and
  satisfies $R_e(\eta) \leq R_e(\eta_\varepsilon) \leq R_e(\eta)$ (where
  the   first   inequality  comes   from   the   fact  that   $R_e$   is
  non-decreasing). Therefore $\eta_\varepsilon$ is  a local maximum with
  $\inf \eta_\varepsilon\geq \varepsilon$, and thus, thanks to the first
  part of  the proof, $\eta_\varepsilon=\un$. This  readily implies that
  $\eta=\un$.   We  deduce that  if  $\eta$  is  a local  maximum,  then
  $\eta=\un$ and thus  it is a global maximum.  This  ends the proof for
  the irreducible case when $R_0>0$.  \medskip

  Recall  that  $R_e$  and  $R_0$  respectively  denote  $R_e[\kk]$  and
  $R_0[\kk]$.  To treat the monatomic  case, recall that for any~$\eta$,
  we know by~\eqref{eq:decomp} that:
  \[
    R_e(\eta) = R_e[\bm{\kk}_\oaa](\bm{\eta}_\oaa),
  \]
  where $\bm{\kk}_\oaa$ (resp.\ $\bm{\eta}_\oaa$)  is the restriction of
  $\kk$     (resp.\      $\eta$)     to     the      atom~$\oa$,     and
  $R_0=R_0[\bm{\kk}_\oaa]>0$.  Let $\eta\in \Deltad$  be a local maximum
  for  $R_e$. Then  $\bm{\eta}_\oaa$  is  a   local  maximum  for
  $R_e[\bm{\kk}_\oaa]$.   We deduce  from the  first part  of the  proof
  applied    to   the    irreducible    kernel   $\bm{\kk}_\oaa$    that
  $\bm{\eta}_\oaa=\ind{\mathrm{a}}$,  and thus  $\eta\geq \ind{\oa}$  as
  well as $R_e(\eta)\geq R_e(\ind{\oa})=R_e(\un)$.  Thus,
  the strategy $\eta$ is a global maximum.  This implies that Assumption
  \ref{hyp:loss-topo}~\ref{hyp:loss**} holds.

  Use that $\ind{\mathrm{a}}$,  the unity function defined  on $\oa$, is
  the only global  maximum of $R_e[\kk_\mathrm{a}]$ thanks  to the first
  part of the proof, to deduce that  $\eta$ is a global maximum of $R_e$
  if and only if $\eta\geq \ind{\oa}$ (in $L^\infty $).
\end{proof}

\subsection{The SIS model}
\label{sec:SIS-hyp} 

The        loss        $\loss=\I$       does        not        satisfies
Assumption~\ref{hyp:loss-topo}~\ref{hyp:loss**} in general even when the
kernel  $\kk=k/\gamma$   is  irreducible   with  $R_0=R_0[\kk]>0$.    Indeed,  by
continuity    of     $R_e$,    there    exists    a     (weak-*)    open
neighborhood~$G$ of $\zero$ such that  $R_e(\eta) < 1$ for all
$\eta\in G$:    consequently   $\I$    is   identically    zero
on~$G$,  and any  $\eta\in G$ is  a local  maximum of
$\loss=\I$. However,  these maxima are  not global in  the super-critical
regime  where~$\I(\un)>0$  (and  $R_0>1$).   For  this  reason,  we  shall
consider           the           following          variant           of
Assumption~\ref{hyp:loss-topo}~\ref{hyp:loss**},  where   one  does  not
consider the zeros of the loss.

\theoremstyle{plain}
\newtheorem*{repeatedAssumption}{Assumption~\ref{hyp:loss-topo}}

\begin{repeatedAssumption}[On the loss]
  \label{hyp:loss=I}
Let  $\loss$ be  a function from $\Deltad$  endowed with 
  the weak-* topology to $\R$.
 \begin{enumerate}[label=(iii')]
 \item \textbf{Maxima.} \label{item:hyp:loss=I}
   Any  local maximum $\eta$ of the loss function $\loss$, such that
 $\loss(\eta)>0$, is a global maximum.
  \end{enumerate}
\end{repeatedAssumption}

We are now ready to check that Assumption~\ref{hyp:loss=I}~\ref{item:hyp:loss=I} holds for the
loss $\loss=\I$ when the kernel $k$ is monatomic. 
Recall $\mathfrak{g}\in \Deltac$ is the maximal equilibrium.

\begin{lemma}\label{lem:L**=I}
  Consider the SIS model $\param=[(\Omega,  \cf, \mu), k, \gamma]$ under
  Assumption~\ref{hyp:k-g} with the loss function $\loss=\I$, and simply write $R_0$ for  $R_0[k/\gamma]$.
  \begin{enumerate}[(i)]
  \item                                                  \label{item:I0}
    Assumption~\ref{hyp:loss-topo}~\ref{hyp:loss+cost}  holds, and  thus
    $\I(\un)>0$, if $R_0>1$.
\item\label{item:I1}
  Assumption~\ref{hyp:loss-topo}~\ref{hyp:loss} holds.

\item\label{item:I2}  If   $k$  is   monatomic  and   $R_0>1$,  then
   Assumption   \ref{hyp:loss=I}~\ref{item:hyp:loss=I}   holds.
Furthermore,   $\eta\in  \Deltad$  is a global maximum of $\I$ if and
only if $\eta\geq \ind{\{ \mathfrak{g}>0\}}$ (in $L^\infty $).

\end{enumerate}
\end{lemma}

\begin{proof}
   Since the loss $\I$ is sub-homogeneous, see
  Proposition~\ref{prop:I}, we deduce from Lemma~\ref{lem:c-dec+L-hom} that
  Assumption~\ref{hyp:loss-topo}~\ref{hyp:loss} holds. 
   Using  Theorem~\ref{th:continuity-I}
  (for the  continuity), Proposition~\ref{prop:I} (for  the monotonicity
  of the function  $\I$), and the fact that  $\I(\un)>0$ if  $R_0>1$,
  see~\eqref{eq:gh>0}, we obtain that the hypothesis
  on  the   loss  in  Assumption~\ref{hyp:loss-topo}~\ref{hyp:loss+cost}
  hold if $R_0>1$.
  \medskip
  
  We  now  prove  Point~\ref{item:I2}.  Assume that  $R_0>1$,  that  is,
  $\I(\un)>0$, and set $\kk=k/\gamma$.  Let $\mathfrak{g}$ be the maximal
  equilibrium  which  is non-zero  as  $R_0>1$.
  Recall that being $\kk'$-invariant depends  only on the support of the
  kernel $\kk'$. Since the kernels $\kk$  and $k$ have the same support,
  and $k$ is monatomic, we deduce  that $\kk$ is monatomic with the same
  atom $\oa$ and  same smallest invariant set containing  $\oa$ given by
  $\{  \mathfrak{g}>0\}$ thanks  to Lemma~\ref{lem:g>0=oa+oi}.   Suppose
  that  $\I$   has  a  local   maximum  at  some   $\eta\in\Deltad$  and
  $\I(\eta)>0$.       For     $\varepsilon\in      (0,     1)$,      set
  $\eta_\varepsilon=  \eta+\varepsilon(\un-\eta)$.    We  have   that  for
  $\varepsilon>0$ small enough:
\begin{equation}
   \label{eq:ineg=I}
  \I(\eta) \geq \I(\eta_\varepsilon)= \int_\Omega \mathfrak{g}_{\eta_\varepsilon}\,
    \eta_\varepsilon \, \mathrm{d}\mu \geq \int_\Omega \mathfrak{g}_{\eta_\varepsilon}\,
    \eta\, \mathrm{d}\mu
    \geq \int_\Omega
    \mathfrak{g}_{\eta}\, \eta\, \mathrm{d}\mu
    = \I(\eta),
 \end{equation}
 where     we    used     that    $\eta\leq     \eta_\varepsilon$    and
 $0\leq  \mathfrak{g}_\eta  \leq  \mathfrak{g}_{\eta_\varepsilon}$,  see
 \eqref{eq:mon-gh}.  Therefore  all these  quantities are  equal.  Since
 $k  \eta_\varepsilon$ and  $k$ have  the same  support, we  deduce that
 $k  \eta_\varepsilon$  is monatomic  with  non-zero  atom $\oa$.   From
 Lemma~\ref{lem:g>0=oa+oi},       we        also       obtain       that
 $\{ \mathfrak{g}_{\eta_\varepsilon}>0\}$ and  $\{ \mathfrak{g}>0\}$ are
 equal, being equal to the  smallest invariant set containing $\oa$.  We
 deduce from~\eqref{eq:ineg=I}, as all  the inequalities are equalities,
 that  $\eta_\varepsilon=\eta$ a.s.\  on  $\{  \mathfrak{g}>0\}$,   and  thus
 $\eta\geq  \ind{\{  \mathfrak{g}>0\}}$ a.s..  Recall  from~\eqref{eq:mon-gh}
 that $\mathfrak{g}_{\eta}\leq  \mathfrak{g}$.  So $\mathfrak{g}_{\eta}$
 is zero outside $\{ \mathfrak{g}>0\}$,  and we deduce that changing the
 value of $\eta$ outside $\{ \mathfrak{g}>0\}$ does not affect the value
 of $\I(\eta)$.  In conclusion, $\eta\in \Deltad$ is a local maximum such
 that $\I(\eta)>0$ if and  only if $\eta\geq \ind{\{ \mathfrak{g}>0\}}$ a.s.,
 and    thus     is    a    global    maximum.
\end{proof}

\section{Technical proofs: properties of~\texorpdfstring{$\I$}{I} and of the maximal
  equilibrium} \label{sec:technical_proofs}

In the SIS model, in order to  stress, if necessary, the dependence of a
quantity $H$, such  as $F_\eta$, $R_e$ or  $\mathfrak{g}_{\eta}$, in the
parameters $k$ and $\gamma$  (which satisfy Assumption~\ref{hyp:k-g}) of
the  model, we  shall  write  $H[k, \gamma]$.  Recall  that  if $k$  and
$\gamma$  satisfy Assumption~\ref{hyp:k-g},  then the  kernel $k/\gamma$
has a finite double norm on $L^p$  for some $p\in (1, +\infty )$ (as
the measure $\mu$ is finite). We now
consider      the     continuity      property      of     the      maps
$\eta       \mapsto       \mathfrak{g}_{\eta}[k,      \gamma]$       and
$(k,\gamma,\eta) \mapsto \mathfrak{g}_{\eta}[k, \gamma]$.
Notice the former function defined on $\Deltac$ is well defined on
$\Deltad$ thanks to~\eqref{eq:mon-gh}.

\begin{lemma}\label{lem:cvgn2}
  Let $((k_n, \gamma_n), n\in \N)$ and $(k,\gamma)$ be kernels and functions satisfying
  Assumption~\ref{hyp:k-g} and $(\eta_n, \, n \in \N)$ be a sequence of elements of
  $\Deltad$ which weak-* converges  to $\eta$.
  \begin{enumerate}[(i)]
  \item\label{lem:cvg2-h} We have $\lim_{n\rightarrow \infty }
    \mathfrak{g}_{\eta_n}[k,\gamma]=\mathfrak{g}_{\eta}[k, \gamma]$ $\mu$-almost surely.
  \item\label{lem:cvg2-kgh} Assume furthermore there exists $p'\in (1,
    +\infty )$ such that $\kk=\gamma^{-1} k$ and $(\kk_n=\gamma^{-1}_n k_n,
    n\in \N)$ have finite double norm on $L^{p'}$ and that $\lim_{n\rightarrow\infty }
    \norm{\kk_n -\kk}_{p',q'}=0$. Then, we have $\lim_{n\rightarrow
    \infty } \mathfrak{g}_{\eta_n}[k_n,\gamma_n]=\mathfrak{g}_{\eta}[k, \gamma]$
    $\mu$-almost surely.
  \end{enumerate}
\end{lemma}

\begin{proof}
  The  proof of  \ref{lem:cvg2-h}  and  \ref{lem:cvg2-kgh} being  rather
  similar, we only  provide the latter and indicate  the difference when
  necessary.           To          simplify,           we          write
  $g_n=\mathfrak{g}_{\eta_n}[k_n,\gamma_n]$.            We           set
  $h_n  =  \eta_ng_n  \in  \Deltad$  for $n\in  \N$.  Since  $\Deltad$  is
  sequentially weak-*  compact, up to  extracting a subsequence,  we can
  assume that  $h_n$ weak-* converges  to a  limit $h\in  \Deltad$. Since
  $F_{\eta_n}[k_n,      \gamma_n](g_n)=0$     for      all     $n\in\N$,
  see~\eqref{eq:F(g)=0}, we have:
  \begin{equation}\label{eq:continuity-I} g_n = \frac{\Tinf_{\kk_n}(\eta_n g_n)}{1+
    \Tinf_{\kk_n}(\eta_n g_n)}
    = \frac{\Tinf_{\kk_n} (h_n)}{1+ \Tinf_{\kk_n}(h_n)} \cdot
  \end{equation}
  We   set    $g=\Tinf_\kk(h)/(1   +   \Tinf_\kk(h))$.     Notice   that
  $ \Tinf_{\kk_n}(h_n)=(\Tinf_{\kk_n}-\Tinf_\kk)(h_n) + \Tinf_\kk(h_n)$.
  We   have    $\lim_{n\rightarrow\infty}   \Tinf_\kk(h_n)=\Tinf_\kk(h)$
  pointwise.                                                       Since
  $\norm{(\Tinf_{\kk_n}-\Tinf_\kk)(h_n)}_{p'}\leq            \norm{\kk_n
    -\kk}_{p',q'}$,  up  to  taking   a  sub-sequence,  we  deduce  that
  $\lim_{n\rightarrow\infty}   (\Tinf_{\kk_n}-\Tinf_\kk)(h_n)=0$  almost
  surely.   (Notice the  previous  step  is not  used  in  the proof  of
  \ref{lem:cvg2-h}             as            $\kk_n=\kk$             and
  $\lim_{n\rightarrow\infty} \Tinf_k(h_n)=\Tinf_k(h)  $ pointwise.) This
  implies that  $g_n$ converges almost  surely to $g$. By  the dominated
  convergence theorem  (recall $\mu$  is finite),  we deduce  that $g_n$
  converges also  in $L^p$ to $g$.   This proves that $h=\eta  g$ almost
  surely. We  get $g=\Tinf_\kk(\eta g)/(1+ \Tinf_\kk(\eta  g))$ and thus
  $F_{\eta}[k,    \gamma](g)=0$ in $\cl$:    $g$    is   an    equilibrium    for
  $F_\eta[k,  \gamma]$.   We  recall from  \cite[Section~3]{ddz-Re}  the
  functional equality $R_e[k'h]=R_e[hk']$, where $k'$ is a kernel, $h$ a
  non-negative  functions such  that the  kernels $k'h$  and $hk'$  have
  finite double norm.  We get:
  \begin{align*}
    R_e[k\eta/\gamma](1- g)= R_e[\kk](\eta(1-g))
    &= \lim_{n\rightarrow\infty }    R_e[\kk_n](\eta_n (1- g_n))\\
   &= \lim_{n\rightarrow\infty }    R_e[k_n\eta_n/\gamma_n](1- g_n)\\
   &\leq  1,
  \end{align*}
 where we used
  the  weak-*  continuity  and  the  stability  of  $R_e$  from  Theorem
  \ref{th:continuity-R}   and  Proposition~\ref{prop:Re-stab}   for  the
  second  equality,   and  Lemma~\ref{lem:Fh>0}~\ref{cor:h=g}   for  the
  inequality.       (Only       the      weak-*       continuity      of
  $\eta'\mapsto   R_e[k/\gamma](\eta')$  is   used  in   the  proof   of
  \ref{lem:cvg2-h} to get  $R_e[k/\gamma](\eta(1-g))\leq 1$.)  Since $g$
  is   an  equilibrium   for   $F_\eta[k,  \gamma]$,   we  deduce   from
  Lemma~\ref{lem:Fh>0}~\ref{cor:h=g}, with $k$ replaced by $k\eta$, that
  $g$      is      the      maximal      equilibrium,      that      is,
  $g =\mathfrak{g}_{\eta}[k,\gamma]$.
\end{proof}

\begin{proof}[Proofs of Theorem \ref{th:continuity-I} and Proposition \ref{prop:I-stab}]
  Under  the  assumptions  of  Lemma~\ref{lem:cvgn2},  taking  the  pair
  $(k_n, \gamma_n)$  equal to $(k,\gamma)$ in  the case \ref{lem:cvg2-h}
  therein,                 we                deduce                 that
  $(\eta_n\,   \mathfrak{g}_{\eta_n}[k_n,\gamma_n],  n\in   \N)$  weak-*
  converges  to  $\eta  \,\mathfrak{g}_{\eta}[k,\gamma]$.  This  implies
  that:
  \[
    \lim_{n\rightarrow \infty } \I[k_n, \gamma_n](\eta_n)
    = \lim_{n\rightarrow \infty } \int_\Omega \eta_n\,
    \mathfrak{g}_\eta[k_n, \gamma_n]\, \mathrm{d}\mu
    = \int_\Omega \eta\,
    \mathfrak{g}_\eta[k, \gamma]\, \mathrm{d}\mu
    = \I[k, \gamma](\eta).
  \]

Taking $(k_n, \gamma_n)=(k, \gamma)$ provides the continuity of $\I[k,
\gamma]$ and thus Theorem \ref{th:continuity-I}. Then, arguing as in the
end of the proof of Proposition \ref{prop:Re-stab}, we get Proposition
 \ref{prop:I-stab}.
\end{proof}

\printbibliography
\end{document}